\documentclass[twoside,journal]{IEEEtran}
\usepackage{amsfonts}
\usepackage{amssymb}
\usepackage{amsmath}
\usepackage{mathrsfs}
\usepackage{verbatim}
\usepackage{subfigure}
\usepackage{balance}
\usepackage{booktabs}
\usepackage{fancybox}
\usepackage{bm}
\usepackage{changepage}
\usepackage{extarrows}
\usepackage{algorithm}
\usepackage{algorithmic}
\usepackage{diagbox}
\usepackage{multirow}
\usepackage{array}
\usepackage{graphicx}
\usepackage{epstopdf}
\usepackage{caption}
\usepackage{color}
\newtheorem{theorem}{Theorem}
\newtheorem{lemma}{Lemma}
\newtheorem{corollary}{Corollary}
\newtheorem{proof}{Proof}
\newtheorem{proposition}{Proposition}

\begin{document}

\title{Physical-Layer Security in Cache-Enabled Cooperative Small Cell Networks Against Randomly Distributed Eavesdroppers}
\author{Tong-Xing~Zheng,~\IEEEmembership{Member,~IEEE,}
Hui-Ming~Wang,~\IEEEmembership{Senior~Member,~IEEE,}
and~Jinhong~Yuan,~\IEEEmembership{Fellow,~IEEE,~}
\thanks{T.-X.~Zheng and H.-M. Wang are with the School of Electronic and Information Engineering,
Xi'an Jiaotong University, Xi'an, Shaanxi 710049, China (zhengtx@mail.xjtu.edu.cn, xjbswhm@gmail.com).
}
\thanks{J. Yuan is with the School of Electrical Engineering and Telecommunications, University of New South Wales, Sydney, NSW 2052, Australia (e-mail: j.yuan@unsw.edu.au).}
}

\maketitle
\vspace{-0.8 cm}

\begin{abstract}
	This paper explores the physical-layer security in a small cell network (SCN) with cooperative cache-enabled small base stations (SBSs) in the presence of {randomly distributed eavesdroppers}. We propose a joint design on the caching placement and the physical-layer transmission to improve the secure content delivery probability (SCDP). We first put forward a hybrid caching placement strategy in which a proportion of the cache unit in each SBS is assigned to store the most popular files (MPFs), while the remaining is used to cache the disjoint subfiles (DSFs) of the less popular files in different SBSs as a means to enhance transmission secrecy and content diversity. 
	We then introduce two coordinated multi-point (CoMP) techniques, namely, joint transmission (JT) and orthogonal transmission (OT), to deliver the MPFs and DSFs, respectively. 
	We derive analytical expressions for the SCDP in each transmission scheme,  considering both non-colluding and colluding eavesdropping scenarios. 
	Based on the obtained analytical results, we jointly design the optimal transmission rates and the optimal caching assignment for maximizing the overall SCDP. Various insights into the optimal transmission and caching designs are further provided. Numerical results are also presented to verify our theoretical findings and to demonstrate the superiority of the proposed caching and transmission strategies.

\end{abstract}


\IEEEpeerreviewmaketitle

\section{Introduction}

\IEEEPARstart{S}{mall} cell network (SCN) is a promising approach to improving network capacity and achieving seamless wireless coverage in the 5G wireless network. Nevertheless, the increasingly dense deployment of SCNs poses a tremendous challenge to the backhaul links and the backhaul capacity has become the major system bottleneck.
To alleviate such bottleneck, wireless caching technique emerges. By pre-storing popular content at the edge of an SCN such as small base stations (SBSs) and reusing the cached content to meet frequent requests from local users, wireless caching is envisioned as an effective solution for relaxing the challenging demand of small cell backhauling and reducing the end-to-end latency \cite{Wang2014Cache}.

{Wireless caching has a serious security vulnerability as any wireless network due to the broadcast nature of the electromagnetic signal propagation. For example, although wireless caching has a great potential to meet the soaring video-on-demand (VoD) streaming traffic in the 5G network \cite{Liu2014Cache,Xiang2018Cache}, the broadcast streaming data by caching nodes are susceptible to potential eavesdroppers such as non-paying subscribers and malicious attackers. 
	It is of great significance to propose security-wise caching schemes that can guarantee both data secrecy and quality of service (QoS). Nonetheless, safeguarding the security for a cache-enabled wireless network is confronted with two major challenges. Specifically, the vast majority of current wireless data services still rely on end-to-end encryption to ensure data secrecy, e.g., the hypertext transfer protocol secure (HTTPS) applied for video streaming applications such as YouTube and Netflix \cite{HTTP}.
	However, such encryption schemes might counteract the benefits of wireless caching in terms of high flexibility and large multiplexing gains since the encrypted content is uniquely defined for each user request and cannot be reused to serve other user requests \cite{Paschos2016Wireless}.} Moreover, using the encryption methods will inevitably introduce a large amount of additional operations in the storage, management and distribution of secret keys, thus degrading the efficiency of content placement and delivery. 
Fortunately, \emph{physical-layer security} (PLS) \cite{Wyner1975Wire-tap}, an information-theoretic approach which has been proven to gain a remarkable secrecy enhancement in various wireless networks \cite{Wang2016Physical_book}-\cite{Liu2017Enhancing}, provides a new opportunity to overcome the above limitations. {PLS achieves wireless secrecy by using the wiretap channel encoding instead of the source encryption such that the cached content still can be reused. Moreover, PLS exploits the randomness inherent to the wireless channels without necessarily relying on secret keys.
All these advantages make PLS and wireless caching easily integrated in a low-complexity and high-flexibility way. 
}

\subsection{Previous Works}
The cache-enabled SCN is first investigated in \cite{Shanmugam2013FemtoCaching}, where wireless caching has been shown to significantly reduce the average downloading delay. In \cite{Liu2014Cache}, \cite{Liu2014Will}, wireless caching has been exploited to improve the energy efficiency of cellular networks. In \cite{Xiang2017Cross}, a joint caching and buffering strategy has been proposed to overcome the backhaul capacity bottleneck and the half-duplex transmission
constraint simultaneously. 
{Considering cooperative transmissions via distributed caching helpers, an optimal caching placement has been designed in \cite{Chae2015Cooperative} as a means to balance the file diversity gain and the cooperation gain. In \cite{Ao2015Distributed,Peng2015Backhaul}, the cache-enabled small cell cooperation in the caching placement (i.e., cache-level cooperation) and the physical-layer transmission (i.e., signal-level cooperation) has been discussed. 
In \cite{Chen2017Cooperative}, the cache- and signal-level cooperation has been leveraged under a combined caching placement scheme to improve the cache service performance and
the energy efficiency.}

Recently, the security issue in the cache-enabled wireless networks has attracted a stream of research, motivated by the coded caching scheme introduced in \cite{Maddah-Ali2014Fundamental}. For example, a coded caching scheme based on Shannon's one-time pad method has been proposed in \cite{Sengupta2015Fundamental} to guarantee information secrecy against eavesdroppers. However, achieving secrecy requires a sufficiently large size of random secret keys, and secure sharing of such massive secret keys will cause a considerable system overhead. This scheme has been extended to the device-to-device networks in \cite{Awan2015Fundamental}, where a sophisticated key generation and encryption scheme has been designed. In \cite{Gerami2015Secure,Gabry2016On}, the security-oriented content placement has been studied based on the maximal distance separable encoding. All these endeavors have dealt with the security
issue from an information-theoretic point of view, but few has considered the characteristics of the physical-layer media.

{PLS in a cache-enabled
	wireless network was not considered until recently. Specifically, the caching-enabled cooperative multi-input multi-output (MIMO) transmission has been first exploited as an effective PLS mechanism to increase the secrecy rate of content delivery, confronting either a single malicious eavesdropper \cite{Xiang2018Cache} or multiple untrusted cache helpers \cite{Xiang2018Secure}.
	Nevertheless, the signaling design therein heavily relies on the estimated instantaneous channel state information (CSI) of the eavesdropper. In practice, it is difficult to estimate such CSI in real time since the eavesdropper usually listens passively.
	Dynamically adjusting the transmission parameters will also increase the system complexity and the end-to-end latency.}
Moreover, the fading feature of the wireless channels and the randomness of the eavesdroppers' locations have significant impacts on the security performance, and meanwhile they can also be exploited to facilitate secure transmissions.
However, these aspects have not yet been investigated in \cite{Xiang2018Cache,Xiang2018Secure}.
{It is worth mentioning that, there has been substantial research on the PLS in random wireless networks with both channel and location uncertainties \cite{Zhou2011Throughput}-\cite{Liu2017Enhancing}.
	Stochastic geometry theory has provided a powerful tool to study the network security performance by modeling the positions of network nodes including the eavesdroppers according to a spatial distribution such as a Poisson point process (PPP) \cite{Haenggi2009Stochastic}.} 
To the best of our knowledge, the potential of PLS in securing the content delivery for a cache-enabled SCN against randomly distributed eavesdroppers is still elusive, and even a fundamental mathematical framework for security performance analysis and optimization from a stochastic geometry perspective is lacking. This has motivated our work.

\subsection{Our Work and Contributions}
In this paper, we will explore the potential of the physical-layer secure transmission in conjunction with caching placement in realizing secure content delivery against randomly distributed eavesdroppers. 
The main contributions of this paper are summarized as follows.
 
\begin{itemize}
	\item We study a combination of cooperations in both cache-level and signal-level. {By cache-level cooperation, we propose a hybrid caching placement strategy based on file partition, where every SBS assigns a proportion of its cache space to store the most popular files (MPFs), while using the remaining to cache the disjoint subfiles (DSFs) of the less popular files as a means to improve content diversity and secrecy. 
		By signal-level cooperation, we put forward two coordinated multi-point (CoMP) techniques, namely, joint transmission (JT) and orthogonal transmission (OT), to deliver the cached MPFs and DSFs, respectively.}
	\item We assess the secure content delivery probability (SCDP), which measures the probability that the reliability and secrecy of content delivery can be guaranteed simultaneously.  
	We provide analytical results for the SCDP in each transmission scheme for both non-colluding and colluding wiretap scenarios. We show that the JT scheme outperforms the OT scheme in terms of transmission reliability, whereas the latter provides a higher level of secrecy.
	
	\item
	We jointly design the optimal transmission rates and the optimal caching assignment proportion to maximize the overall SCDP under the proposed caching and transmission strategies.
	The whole maximization procedure is decomposed into two phases.  First, the optimal transmission rates that globally maximize the SCDP in the JT and OT schemes are obtained by solving a scalar convex problem and by addressing a vector convex problem through an alternating optimization, respectively; subsequently, the optimal caching assignment  proportion that maximizes the overall SCDP is derived in a closed-form expression. Various insights into the optimal transmission and caching designs  are further provided. 
\end{itemize}

\subsection{Organization and Notations}
The remainder of this paper is organized as follows. In Section II, we describe the system model and the underlying optimization problem. In Section III, we analyze
the connection probability and the secrecy probability for both the JT and OT schemes. In Section IV, we design the optimal transmission rates and the optimal caching assignment proportion to maximize the overall SCDP. In Section V, we present numerical results to validate our theoretical analysis. In Section VI, we conclude our work and provide several future research directions.

\emph{Notations}: Bold lowercase letters denote column vectors. $(\cdot)^{\rm T}$, $|\cdot|$, $\|\cdot\|_1$, $\lfloor \cdot\rfloor$, $\mathbb{P}\{\cdot\}$, $\mathbb{E}_A[\cdot]$ denote the operations of transpose, absolute value, L-1 norm, round down, probability and mathematical expectation taken over the random variable (RV) $A$, respectively.
$\Gamma(z)= \int_0^{\infty}e^{-t}t^{z-1}dt$, $\Gamma(a,x) = \int_x^{\infty}e^{-t}t^{a-1}dt$ and ${\rm erf}(z)={2}/{\sqrt{\pi}}\int_0^ze^{-t^2}dt$ are the gamma function \cite[Eqn. (8.310.1)]{Gradshteyn2007Table},
the incomplete gamma function \cite[Eqn. (8.350.2)]{Gradshteyn2007Table}, and the error function \cite[Eqn. (8.250.1)]{Gradshteyn2007Table}, respectively. ${S}\setminus  s_k$ denotes the set obtained by canceling the subset $s_k$ from the superset ${S}$. 

\section{System Model}
\begin{figure}[!t]
	\centering
	\includegraphics[width = 3.5in]{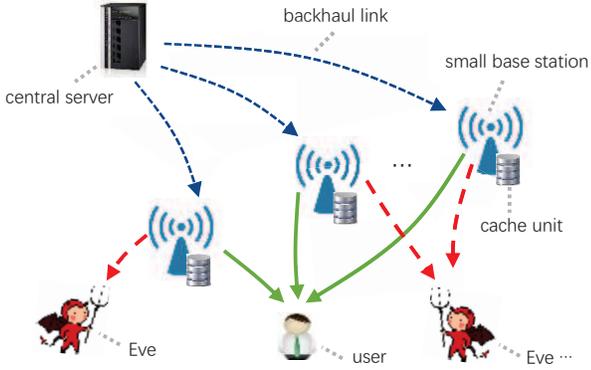}
	\caption{Illustration of a cache-enabled cooperative SCN. The ongoing content delivery from the SBSs to an intended user is overheard by randomly located eavesdroppers.}
	\label{SYSTEM_MODEL}
\end{figure}
We consider security issues in a cache-enabled SCN where a set of SBSs, $\mathcal{K}=[1,\cdots,K]$, cooperatively deliver paid content to a subscriber coexisting with randomly located eavesdroppers (e.g., non-paying subscribers). 
All the SBSs are connected to a central server located at the edge of the core network via wired backhaul links. A cache unit is deployed at each SBS for pre-fetching popular content during off-peak hours from the central server. Once a file requested  by a local user is pre-stored in the cache units, i.e., a \emph{cache hit} event takes place, the file can be delivered to the user by the SBSs directly. Otherwise, a \emph{cache miss} event is deemed to occur and the SBSs will fetch the requested file from the central server before conveying it to the user. We assume that all the network nodes including the SBSs, legitimate users and eavesdroppers each have a single antenna, and only one user can be served in each time slot.
{The locations of the SBSs and of the users are assumed deterministic, while those of eavesdroppers are random and unknown.}
Without loss of generality, we place a typical user at the origin of the polar coordinate, and denote the locations of the $k$-th nearest SBS and of the $j$-th nearest eavesdropper to the typical user as  $\{b_{k}:(r_{b,k},\theta_{b,k})\}$ and $\{e_j:(r_{e,j},\theta_{e,j})\}$, with $r$ and $\theta$ being the corresponding distance and angle, respectively. The distance between the $k$-th SBS and the $j$-th eavesdropper is expressed as $r_{j,k}=\sqrt{r_{b,k}^2+r_{e,j}^2-2r_{b,k}r_{e,j}\cos(\theta_{b,k}-\theta_{e,j})}$. Since the eavesdroppers are possibly randomly located over the entire network, we model their positions as a stationary PPP $\Phi_e$ on the two-dimensional plane $\mathbb{R}^2$ with density $\lambda_e$, i.e., $e_j\in\Phi_e$ \cite{Zhou2011Throughput}-\cite{Liu2017Enhancing}.  

{
To model the wireless channels, including the main channels spanning from the SBSs to the typical user and the wiretap channels spanning from the SBSs to the eavesdroppers, we consider a standard distance-based path loss governed by the exponent $\alpha$ along with Rayleigh fading.
Hence, the channel gains from the $k$-th SBS to the typical user and to the $j$-th eavesdropper can be expressed as $h_{b,k}r_{b,k}^{-\alpha/2}$ and $h_{j,k}r_{j,k}^{-\alpha/2}$, respectively, where $h_{b,k}$ and $h_{j,k}$ denote the independent and identically distributed (i.i.d.) fading coefficients obeying the circularly symmetric complex Gaussian distribution with zero mean and unit variance.
We assume that the SBSs only know the statistic CSI (as opposed to the instantaneous) of the main and the wiretap channels.}\footnote{This is a generic assumption in literature on PLS \cite{Wang2016Physical_book}-\cite{Liu2017Enhancing}. In particular, if the potential eavesdroppers are also regular users but the ongoing content delivery should be kept secret to them, it is possible for the SBSs to acquire the statistic CSI of the eavesdroppers and the distribution of their locations through a large amount of information exchange during other time slots.}

\subsection{Caching Placement Strategy}
The central server owns a library of $N$ equal-size files with $F_n$ being the $n$-th most popular file. We assume that users make request for files independently with some probabilities according to a commonly adopted popularity pattern, i.e., Zipf distribution \cite{Shanmugam2013FemtoCaching}. With the Zipf distribution, the request probability for the $m$-th most popular file is given by  
\begin{equation}\label{zipf}
f_m = \frac{m^{-\gamma}}{\sum_{n=1}^{N}n^{-\gamma}},
\end{equation}
where $\gamma$ models the skewness of the popularity profile. In particular, the popularity profile is uniform over files for $\gamma=0$, and becomes
more concentrated as $\gamma$ increases.

We consider a finite caching capacity that the $K$ SBSs each can store up to $L$ files and the total capacity $KL$ is less than the library size $N$, i.e., $KL<N$.
To make efficient use of the cache units, we propose a hybrid caching placement strategy to distribute files of different popularities to the SBSs, where each file is partitioned into $K$ equal-size subfiles. 
To be specific, we divide the cache unit in every SBS into two portions, where a proportion $\phi$ of the cache unit is used for storing the most popular files (i.e., MPF caching) and the rest proportion $1-\phi$ is reserved for caching the disjoint subfiles of the less popular files in different SBSs (i.e., DSF caching). 
Under such a hybrid caching strategy, file $F_n$ with $1\le n\le\lfloor \phi L\rfloor $ belongs to the MPF group where all the SBSs  possess a copy of it; file $F_n$ with $\lfloor \phi L\rfloor < n \le \lfloor \phi L\rfloor + K(L-\lfloor \phi L\rfloor)$ belongs to the DSF group where the SBSs each store a unique subfile of it; file $F_n$ with $ n > \lfloor \phi L\rfloor + K(L-\lfloor \phi L\rfloor)$ is not cached and only can be fetched from the central server.

{It is worth noting that, the DSF caching mode helps to enhance  content diversity, i.e., allowing more files to be cached with a finite caching capacity. This can increase the cache hit probability and avoid frequently fetching content from the central server, thus making the communication more efficient. 
Furthermore, the different subfiles of each file are elaborately encoded such that one only can decode them in a sequential manner \cite{Ohm2005Advances}. Hence, the risk of the entire file being illegally intercepted can be significantly lowered since it is difficult to decode all the subfiles simultaneously. In other words, content diversity can guarantee a higher level of secrecy.}

\subsection{Cooperative Transmission Schemes}
When the SBSs receive a file request, they adopt different transmission schemes depending on the popularity of the requested file. We propose two CoMP schemes, namely, JT and OT schemes, which are described as below. 

\subsubsection{Joint Transmission (JT)}
When the requested file $F_n$ belongs to the MPF group, the SBSs each have a copy of it. In order to enhance transmission reliability, the SBSs deliver the same file to the user simultaneously. We refer to this scheme as the JT scheme.\footnote{{Without the instantaneous CSI of the main channels, the JT scheme corresponds to a non-coherent multi-point joint transmission. If the instantaneous CSI is available, a coherent joint transmission or distributed beamforming can be realized to further improve transmission reliability, which however will increase the system overhead.}} The received signal-to-noise ratios (SNRs) at the user and at the $j$-th eavesdropper respectively can be given by
\begin{equation}\label{snr_o_jt}
{\rm SNR}_b = \rho\left|\sum_{k=1}^K h_{b,k}r_{b,k}^{-\alpha/2}\right|^2,
\end{equation}
\begin{equation}\label{snr_e_jt}
{\rm SNR}_j = \rho\left|\sum_{k=1}^K h_{j,k}r_{j,k}^{-\alpha/2}\right|^2,~~~\forall e_j\in\Phi_e,
\end{equation}
{where $\rho = {P}/{(WN_0)}$ denotes the normalized SNR, with $P$, $W$ and $N_0$ being the SBS transmit power, the total available bandwidth and the noise spectral density, respectively.} We consider identical noise spectral density at all the receivers. 
\subsubsection{Orthogonal Transmission (OT)}
When the requested file $F_n$ falls within the DSF group, the SBSs have different subfiles of it. In order
to avoid the co-channel interference, the SBSs use orthogonal
frequency spectrum to deliver these subfiles, i.e., each SBS occupies $1/K$ of the overall bandwidth. We name this scheme the OT scheme. For the subfile delivered from the $k$-th SBS, the received SNRs at the user and at the $j$-th eavesdropper respectively can be given by
\begin{equation}\label{snr_o_ot}
{\rm SNR}_{b,k} =  {K\rho}\left|h_{b,k}\right|^2r_{b,k}^{-\alpha},
\end{equation}
\begin{equation}\label{snr_e_ot}
{\rm SNR}_{j,k} =  {K\rho}\left|h_{j,k}\right|^2r_{j,k}^{-\alpha},~~~\forall e_j\in\Phi_e.
\end{equation}
Note that the factor $K$ exists due to a $1/K$ decrement of the
available bandwidth for each SBS.

\subsubsection{Cache Miss (CM)}
When the requested file $F_n$ is not cached by the SBSs, i.e., a cache miss event occurs, all the SBSs fetch the requested file from the central server and then deliver it to the user simultaneously. Hence, the received SNRs at the user and at the $j$-th eavesdropper share the same expressions as \eqref{snr_o_jt} and \eqref{snr_e_jt}, respectively. For ease of statement, we call this scheme the CM scheme. The wired backhauling process is assumed to be secure, whereas it causes extra delivery delay compared with the cache hit case, which will result in a lower  end-to-end rate. This negative impact will be considered in the subsequent performance analysis and optimization.

\subsection{Performance Metrics}
To securely deliver the content, the well-known Wyner's wiretap encoding scheme is adopted to encode the confidential information, where redundant information is intentionally embedded to confuse eavesdroppers \cite{Wyner1975Wire-tap}. The transmission rates of
the confidential information and the redundant information are
referred to as the secrecy rate $R_s$ and the redundant rate $R_e$,
and the wiretap codeword rate is $R_t = R_s + R_e$. 
{We denote the achievable rates of the main channel and the wiretap channel as $C_b=\log_2(1+{\rm SNR}_b)$ and $C_e=\log_2(1+{\rm SNR}_e)$, respectively, with ${\rm SNR}_b$ and ${\rm SNR}_e$ being the corresponding SNRs.\footnote{{All the rate parameters in this paper are measured in the unit: bits/s/Hz.}} If $C_b$ is larger than or equal to $R_t$, the legitimate user can recover the confidential information. The probability that this event happens is called the connection probability, which is defined as $p_c \triangleq \mathbb{P}\{C_b\ge R_t\}$ \cite{Wang2016Physical}.
If $C_e$ is lower than or equal to $R_e$, no confidential information will be leaked to the eavesdropper. 
The probability of this event occurs is termed the secrecy probability, which is defined as $p_s \triangleq \mathbb{P}\{C_e\le R_e\}$ \cite{Wang2016Physical}. }

The content delivery is secure only if the reliability (connection) and the secrecy are guaranteed simultaneously. In this paper, we employ the metric named SCDP to quantify the probability that a secure content delivery event occurs. For a specific transmission scheme, the SCDP can be defined as the product of the connection probability and the secrecy probability, i.e., $\mathcal{P}^{\rm S}_{scd} \triangleq p^{\rm S}_{c}p^{\rm S}_{s}$ for $\rm S\in\{ JT,OT,CM\}$. Therefore, the overall SCDP under the proposed caching and transmission strategies can be expressed as 
\begin{equation}\label{def_overall_psct}
\mathcal{P}_{scd} = \sum_{\rm S\in\{ JT,OT,CM\}}p_{tr}^{\rm S}(\phi)~ \mathcal{P}^{\rm S}_{scd},
\end{equation}
where $p_{tr}^{\rm S}(\phi)$ denotes the probability of the scheme $\rm S\in\{ JT,OT,CM\}$ being adopted for content delivery, which can be calculated from \eqref{zipf}, i.e., 
\begin{subequations}
	\begin{align}\label{p_jt}
	&p_{tr}^{\rm JT}(\phi) = \sum_{n=1}^{\lfloor{\phi L}\rfloor}f_n,~~p_{tr}^{\rm OT}(\phi) = \sum_{\lfloor{\phi L}\rfloor+1}^{\lfloor \phi L\rfloor + K(L-\lfloor \phi L\rfloor)}f_n,
	\\\label{p_cm}~&p_{tr}^{\rm CM}(\phi) =1 -  p_{tr}^{\rm JT}(\phi) - p_{tr}^{\rm OT}(\phi).
	\end{align}
\end{subequations}

We emphasize that the redundant rate $R_e$ (for a target secrecy rate $R_s$) and the caching assignment proportion $\phi$ play critical roles in increasing the SCDP. Specifically, $R_e$ triggers a trade-off between transmission reliability and secrecy.
Choosing a larger $R_e$ increases the secrecy probability $p_s$ but decreases the connection probability $p_c$ as the wiretap codeword rate $R_t$ also increases.
Likewise, $\phi$ strikes a non-trivial trade-off between transmission reliability, secrecy and content diversity. On one hand, devoting a larger proportion $\phi$ for MPFs increases the probability of adopting the JT scheme and thus enhancing transmission reliability. On the other hand, assigning a larger proportion for DSFs (i.e., a smaller $\phi$) increases the probability of using the OT scheme and thus improving transmission secrecy and content diversity. The overall balance of these opposite impacts on the SCDP should be carefully addressed. 

In the following sections, we will first derive the connection probability $p_c$ and the secrecy probability $p_s$ for various transmission schemes, and then we will design the optimal redundant rate $R_e$ and the optimal caching assignment proportion $\phi$ to maximize the overall SCDP $\mathcal{P}_{scd}$.

\section{Connection and Secrecy Probabilities Analysis}
Since the SCDP is represented as the product of the connection probability $p_c$ and the secrecy probability $p_s$, this section will analyze $p_c$ and $p_s$ for the JT and OT schemes, respectively. 
The analysis for the CM scheme is similar as that for the JT scheme, and the only difference lies in the backhaul delay caused by the content fetching process in the former. Due to the extra delay, the actual delivery time in the CM scheme reduces and hence the required secrecy rate increases from $R_s$ (for the JT scheme) to $\delta R_s$, where $\delta>1$ captures the impact of backhaul delay. 

{It is worth mentioning that, for a fair comparison between the JT and OT schemes, we consider identical target secrecy rate $R_s$ in both schemes. In practice, the secrecy rate can correspond to the end-to-end rate of a specific service requested by subscribers such that its value may be pre-established.}
Furthermore, all the SBSs should set a same wiretap codeword rate $R_t$ and a same redundant rate $R_e$ in the JT scheme due to the joint transmission of a same file; whereas they can choose different rates in the OT scheme due to the orthogonal transmission of different subfiles. 
{This also means that, the wiretap codeword rates (also the redundant rates) in the JT and OT schemes are not necessarily the same and actually their values should be properly designed for maximizing the SCDP as will be discussed in Sec. IV.} 

\subsection{Connection Probability}

\subsubsection{JT Scheme}
In this case, the achievable rate of the main channel is $C_b = \log_2\left(1+{\rm SNR}_b\right)$ with  ${\rm SNR}_b$ given in \eqref{snr_o_jt}. Let $R_t = \log_2\left(1+\beta_t\right)$ be the wiretap codeword rate with $\beta_t\triangleq 2^{R_t}-1$. The connection probability can be expressed as $p_c^{\rm JT} = \mathbb{P}\{{\rm SNR}_b\ge \beta_t\}$, which is calculated as 
\begin{equation}\label{pc_jt}
p_c^{\rm JT} =\mathbb{P}\left\{\left|\sum_{k=1}^K h_{b,k}r_{b,k}^{-\alpha/2}\right|^2\ge \frac{\beta_t}{\rho}\right\}\stackrel{\mathrm{(a)}}=e^{-\frac{\beta_t/\rho}{\sum_{k=1}^Kr_{b,k}^{-\alpha} }},
\end{equation}
where step $\rm (a)$ follows from the fact that $\left|\sum_{k=1}^K h_{b,k}r_{b,k}^{-\alpha/2}\right|^2$ is exponentially distributed with mean ${\sum_{k=1}^Kr_{b,k}^{-\alpha}}$.

\subsubsection{OT Scheme}
In this case, the achievable rate of the main channel of the $k$-th SBS is $C_{b,k}=\log_2\left(1+{\rm SNR}_{b,k}\right)$ with ${\rm SNR}_{b,k}$ given in \eqref{snr_o_ot}. Let $R_{t,k}=\log_2\left(1+\beta_{t,k}\right)$ be the wiretap codeword rate of the $k$-th SBS {with $\beta_{t,k}=2^{R_{t,k}}-1$}. 
Since the entire file can be recovered only if all the subfiles have already been decoded. Therefore, the connection probability can be expressed as $p_c^{\rm OT} =\mathbb{P}\left\{\bigcap_{k\in\mathcal{K}}{\rm SNR}_{b,k}\ge \beta_{t,k}\right\} $, which is calculated as 
\begin{equation}\label{pc_ot}
p_c^{\rm OT}=\prod_{k=1}^K\mathbb{P}\left\{\left| h_{b,k}\right|^2r_{b,k}^{-\alpha}\ge \frac{\beta_{t,k}}{K\rho}\right\}= e^{-\frac{\sum_{k=1}^Kr_{b,k}^{\alpha}\beta_{t,k}}{K\rho}}.
\end{equation}

Comparing \eqref{pc_ot} with \eqref{pc_jt}, we find that if $\beta_t=\beta_{t,k}$ for $k\in\mathcal{K}$, i.e., the same wiretap codeword rate is used in the JT and OT schemes, we have $p_c^{\rm JT}\ge p_c^{\rm OT}$ as ${1}/{ \sum_{k=1}^Kr_{b,k}^{-\alpha}}\le\min_{k\in\mathcal{K}}r_{b,k}^{\alpha}\le\sum_{k=1}^Kr_{b,k}^{\alpha}/K$. This shows the superiority of the JT scheme in terms of transmission reliability.

\subsection{Secrecy Probability}
We consider both non-colluding and colluding eavesdropping (NCE/CE) scenarios. In the NCE case, eavesdroppers individually decode the confidential information and thus the content can be delivered secretly if only the achievable rate of the most deteriorate eavesdropper $\max_{e_j\in\Phi_e}C_{j}$ does not exceed the redundant rate $R_e$, i.e., $\max_{e_j\in\Phi_e}C_{j}\le R_e$. In the CE case, eavesdroppers jointly decode the confidential information using the maximal ratio combination (MRC) method. The content delivery is deemed to be secret only if the achievable rate of the equivalent wiretap channel $C_{e}$ does not lie beyond the redundant rate $R_e$, i.e., $C_{e}\le R_e$. 

\subsubsection{JT Scheme for NCE Case}
In this case, the achievable rate of the $j$-th eavesdropper is $C_{j}=\log_2\left(1+{\rm SNR}_{j}\right)$ with ${\rm SNR}_{j}$ given in \eqref{snr_e_jt}. Let $R_{e} = \log_2\left(1+\beta_{e}\right)$ be the redundant rate with $\beta_e\triangleq2^{R_e}-1$. The secrecy probability can be expressed as $p_{s,nce}^{\rm JT}=\mathbb{P}\left\{\max_{e_j\in\Phi_e}{\rm SNR}_{j}\le \beta_e\right\}$, which is calculated as 
\begin{align}\label{ps_jt_nce}
p_{s,nce}^{\rm JT} &
=\mathbb{E}_{\Phi_e}\left[\prod_{e_j\in\Phi_e}\mathbb{P}\left\{\left|\sum_{k=1}^K h_{j,k}r_{j,k}^{-\alpha/2}\right|^2\le\frac{\beta_e}{\rho}\right\}\right]\nonumber\\
& \stackrel{\mathrm{(b)}}=\mathbb{E}_{\Phi_e}\left[\prod_{e_j\in\Phi_e}\left(1-e^{-\frac{\beta_e/\rho}{\sum_{k=1}^Kr_{k}^{-\alpha} }}\right)\right]\nonumber\\
&
\stackrel{\mathrm{(c)}}=\exp\left(-2\lambda_e\int_0^\infty\int_0^{\pi}e^{-\frac{\beta_e/\rho}{\sum_{k=1}^Kr_{k}^{-\alpha} }}rdrd\theta\right),
\end{align}
where step $\rm (b)$ follows from knowing that $\left|\sum_{k=1}^K h_{j,k}r_{j,k}^{-\alpha/2}\right|^2$ is an exponential RV with mean ${\sum_{k=1}^Kr_{j,k}^{-\alpha}}$, and step $\rm (c)$ holds by using the probability generating functional (PGFL) over a PPP \cite{Chiu2013Stochastic} with
$r_{k}=\sqrt{r_{b,k}^2+r^2-2r_{b,k}r\cos(\theta_{b,k}-\theta)}$. Although the result in \eqref{ps_jt_nce} does not appear in a closed form, the integral therein is fairly easy to compute.

\subsubsection{OT Scheme for NCE Case}
In this case, the achievable rate of the wiretap channel from the $k$-th SBS to the $j$-th eavesdropper is $C_{j,k}=\log_2\left(1+{\rm SNR}_{j,k}\right)$ with ${\rm SNR}_{j,k}$ given in \eqref{snr_e_ot}. Let $R_{e,k}=\log_2\left(1+\beta_{e,k}\right)$ be the redundant rate of the $k$-th SBS {with $\beta_{e,k}\triangleq2^{R_{e,k}}-1$}. Note that a file is intercepted only if all its subfiles have already been intercepted. Then, the secrecy probability can be given by $p_{s,nce}^{\rm OT}
=\mathbb{P}\left\{\bigcup_{k\in\mathcal{K}}{\rm SNR}_{j,k}\le \beta_{e,k}, \forall e_j\in\Phi_e\right\}$, which is calculated as 
\begin{align}\label{ps_ot_nce}
p_{s,nce}^{\rm OT}&
=\mathbb{E}_{\Phi_e}\left[\prod_{e_j\in\Phi_e}\left(1-\prod^K_{k=1}\mathbb{P}\left\{{\rm SNR}_{j,k}> \beta_{e,k} \right\}\right)\right]\nonumber\\
&=\mathbb{E}_{\Phi_e}\left[\prod_{e_j\in\Phi_e}\left(1-e^{-\frac{\sum_{k=1}^Kr_{k}^{\alpha}\beta_{e,k}}{K\rho }}\right)\right]\nonumber\\
&=\exp\left(-2\lambda_e\int_0^\infty\int_0^{\pi}e^{-\frac{\sum_{k=1}^Kr_{k}^{\alpha}\beta_{e,k}}{K\rho }}rdrd\theta\right).
\end{align}

Comparing \eqref{ps_ot_nce} with \eqref{ps_jt_nce}, if the same redundant rate is employed in the JT and OT schemes, i.e., $\beta_e = \beta_{e,k}$ for $k\in\mathcal{K}$, we have $p_{s}^{\rm OT}\ge p_{s}^{\rm JT}$ since ${1}/{ \sum_{k=1}^Kr_{k}^{-\alpha}}\le\min_{k\in\mathcal{K}}r_{k}^{\alpha}<\sum_{k=1}^Kr_{k}^{\alpha}/K$. This means the OT scheme provides a higher level of secrecy than does the JT scheme.
For the special case with $\alpha = 2$, we can obtain a more concise expression for $p_{s,nce}^{\rm OT}$ as given below
\begin{equation}
p_{s,\alpha=2}^{{\rm OT}} = \exp\left(-{\rho\lambda_e}\left({\pi}+{{Z}_K}\right)e^{-\frac{\sum_{k=1}^Kr_{b,k}^2\beta_{e,k}}{K\rho}}\right),
\end{equation}
where ${Z}_K=\int_0^{\pi}{\sqrt{\pi}z}e^{z^2}\left[1+{{\rm erf}(z)}\right]d\theta$ with $z = \sum_{k=1}^Kr_{b,k}\cos(\theta_{b,k}-\theta)/({K\sqrt{\rho}})$.

\subsubsection{JT Scheme for CE Case}
In this case, the achievable rate of the equivalent wiretap channel is $C_e = \log_2\left(1+\sum_{e_j\in\Phi_e}{\rm SNR}_{j}\right)$ with ${\rm SNR}_{j}$ given in \eqref{snr_e_jt}. Let $I_e = \sum_{e_j\in\Phi_e}{\rm SNR}_{j}$. Then, the secrecy probability can be expressed as the cumulative distribution function (CDF) of $I_e$, i.e., $p_{s,ce}^{\rm JT} = \mathbb{P}\left\{I_e\le \beta_e\right\}$. In order to compute $p_{s,ce}^{\rm JT}$, we first calculate the Laplace transform of $I_e$.
\begin{lemma}\label{lemma_1}
	The Laplace transform of $I_e$ evaluated at value $s$ is given by
	\begin{equation}\label{laplace_lemma}
	\mathcal{L}_{I_e}(s)  =\exp\left(-2\lambda_e\int_0^{\infty}\int_0^{\pi}\frac{s\rho\sum_{k=1}^Kr_{k}^{-\alpha}}{1+s\rho\sum_{k=1}^Kr_{k}^{-\alpha}}rdrd\theta\right),
	\end{equation}
	where $r_k$ shares the same expression as \eqref{ps_jt_nce}.
\end{lemma}
\begin{proof}
	Recalling \eqref{snr_e_jt}, the Laplace transform $\mathcal{L}_{I_e}\left(s\right) = \mathbb{E}_{I_e}\left[e^{-sI_e}\right]$ can be calculated as 
	\begin{align}\label{laplace}
	\mathcal{L}_{I_e}(s)  &\stackrel{\mathrm{(d)}}=\mathbb{E}_{\Phi_e}\left[\prod_{e_j\in\Phi_e}\mathbb{E}_{h_{j,k}}\left[e^{-s\rho\left|\sum_{k=1}^K h_{j,k}r_{j,k}^{-\alpha/2}\right|^2}\right]\right]
	\nonumber\\
	&\stackrel{\mathrm{(e)}}=\mathbb{E}_{\Phi_e}\left[\prod_{e_j\in\Phi_e}\frac{1}{1+s\rho\sum_{k=1}^Kr_{j,k}^{-\alpha}}\right],
	\end{align}
	where step $\rm{(d)}$ follows from the independence between channel fading and the PPP such that the expectation over $h$ can be moved inside the product; step $\rm (e)$ is obtained by calculating the Laplace transform of an exponential RV $\left|\sum_{k=1}^K h_{j,k}r_{j,k}^{-\alpha/2}\right|^2$. Applying the PGFL over a PPP with \eqref{laplace} completes the proof.
\end{proof}

It is intractable to give a closed form for the exact $p_{s,ce}^{\rm JT}$. In the following theorem, we resort to a widely used approximation method \cite{Singh2015Tractable} and derive a closed-form approximation for  $p_{s,ce}^{\rm JT}$.
\begin{theorem}\label{theorem_1}
	The secrecy probability in the JT scheme for the CE case satisfies
\begin{equation}\label{ps_jt_ce}
p_{s,ce}^{\rm JT} \lessapprox  \sum_{m=1}^M {M\choose m}(-1)^{m+1}\mathcal{L}_{I_e}\left(\frac{m\xi}{\beta_e}\right),
\end{equation}
where $\xi \triangleq M(M!)^{-1/M}$ and $M$ is the number of terms used in the approximation.
\end{theorem}
{\begin{proof}
The secrecy probability $p_{s,ce}^{\rm JT} = \mathbb{P}\left\{I_e\le \beta_e\right\}$ can be calculated as follows,
	\begin{align}\label{ps_jt_ce_proof}
	p_{s,ce}^{\rm JT} 
&	=\mathbb{P}\left\{{I_e}/{\beta_e}\le 1 \right\}\stackrel{\mathrm{(f)}}\approx
	\mathbb{P}\left\{{I_e}/{\beta_e}\le \iota \right\}\nonumber\\
	&	\stackrel{\mathrm{(g)}}\lessapprox 1 - \mathbb{E}_{I_e}\left[\left(1 - e^{-{\xi I_e}/{\beta_e}}\right)^M\right],
	\end{align}
	{where the dummy variable $\iota$ in step $\rm (f)$ is a normalized gamma RV with the PDF $f_{\iota}(x)=x^{M-1}e^{-x}/\Gamma(M)$, and this step follows from the fact that $\iota$ converges to identity as $M$ approaches infinity \cite{Singh2015Tractable}; step $\rm (g)$ yields an upper bound by invoking Alzer's inequality \cite{Alzer1997On}, i.e., $\mathbb{P}\{\iota\ge z\}\le 1 - \left[1-e^{-\xi z}\right]^M$ for a constant $z>0$.} Using the binomial expansion with \eqref{ps_jt_ce_proof} and substituting in the Laplace transform $\mathcal{L}_{I_e}\left(s\right) = \mathbb{E}_{I_e}\left[e^{-sI_e}\right]$ with $s={m\xi}/{\beta_e}$ completes the proof.
\end{proof}}

\subsubsection{OT Scheme for CE Case}
In this case, the achievable rate of the equivalent wiretap channel from the $k$-th SBS to the colluding eavesdroppers is $C_{e,k}=\log_2\left(1+\sum_{e_j\in\Phi_e}{\rm SNR}_{j,k}\right)$ with ${\rm SNR}_{j,k}$ given in \eqref{snr_e_ot}. The secrecy probability can be interpreted as the complement of the probability that all the subfiles are intercepted by the eavesdroppers, which is expressed as 
\begin{equation}\label{ps_ot_ce}
p_{s,ce}^{\rm OT} = 1 - 
\mathbb{E}_{\Phi_e}\left[\prod_{k=1}^K\mathbb{P}\left\{\sum_{e_j\in\Phi_e}|h_{j,k}|^2r_{j,k}^{-\alpha}>\frac{\beta_{e,k}}{K\rho}\right\}\right].
\end{equation}

 {Since the locations of the SBSs are deterministic, the distances between the $j$-th eavesdropper and any two SBSs actually are not independent \cite{Zhang2017Energy}.} 
Hence, the expectation over the PPP $\Phi_e$ in \eqref{ps_ot_ce} cannot be moved inside the product, which makes $p_{s,ce}^{\rm OT}$ difficult to compute. To facilitate the calculation, we first consider a disc $\mathcal{B}(o,R)$ centered at the origin $o$ with a radius $R$ and let $\Phi_{e}^R \triangleq \Phi_e\cap\mathcal{B}_{o,R}$ denote the location set of the eavesdroppers residing in the disc $\mathcal{B}_{o,R}$. We then calculate the inner probability in \eqref{ps_ot_ce} resorting to a common gamma approximation \cite{Heath2013Modeling}. Specifically, we approximate the term $X_k = \sum_{e_j\in\Phi_{e}^R}|h_{j,k}|^2r_{j,k}^{-\alpha}$ as a gamma RV, the probability density function (PDF) of which is given by
\begin{equation}\label{pdf_gamma}
f_{X_k}\left(x_k;\upsilon_k,\tau_k\right)=\frac{x_k^{\upsilon_k-1}e^{-\frac{x_k}{\tau_k}}}{\tau_k^{\upsilon_k}\Gamma(\upsilon_k)}.
\end{equation}
The parameters $\upsilon_k$ and $\tau_k$ can be derived from matching the first and second moments of $X_k$, which are given in the following lemma with the detailed calculation relegated to Appendix \ref{appendix_lemma_2}. 
\begin{lemma}\label{lemma_2}
	For fixed positions of eavesdroppers in the disc $\mathcal{B}(o,R)$, $\upsilon_k$ and $\tau_k$ are given by
	\begin{equation}\label{upsilon_k}
	\upsilon_k= \frac{\left(\sum_{e_i\in\Phi_{e}^R}r_{i,k}^{-\alpha}\right)^2}{\sum_{e_j\in\Phi_{e}^R}r_{j,k}^{-2\alpha}},~\tau_k = \frac{\sum_{e_i\in\Phi_{e}^R}r_{i,k}^{-2\alpha}}{\sum_{e_j\in\Phi_{e}^R}r_{j,k}^{-\alpha}}.
	\end{equation}
\end{lemma}

For a PPP, the probability of having $J$ eavesdroppers inside the disc $\mathcal{B}_{o,R}$ is given by \cite{Chiu2013Stochastic}
\begin{equation}\label{ppp_n_j}
\mathcal{O}_J\triangleq\mathbb{P}\{n=J\} = e^{-\pi\lambda_e R^2}\frac{\left(\pi \lambda_e R^2\right)^J}{J!}.
\end{equation}
Hence, the secrecy probability in \eqref{ps_ot_ce} can be rewritten as 
\begin{equation}\label{ps_equal}
p_{s,ce}^{\rm OT} = 1-\lim_{R\rightarrow\infty}\sum_{J=1}^{\infty}\mathcal{O}_J\mathbb{E}_{\Phi_{e}^R}\left[\prod_{k=1}^K\mathbb{P}\left\{X_k>{\frac{\beta_{e,k}}{K\rho}}\right\}\bigg|{\Phi_{e}^R},J\right].
\end{equation}
The exact $p_{s,ce}^{\rm OT}$ is provided by the following theorem.
\begin{theorem}\label{theorem_2}
	The secrecy probability in the OT scheme for the CE case is given by 
	\begin{align}\label{ps_oft}
	p_{s,ce}^{\rm OT}  = 1-&\lim_{R\rightarrow\infty}\sum_{J=1}^{\infty}\frac{\lambda_e^Je^{-\pi\lambda_e R^2}}{J!}\times\nonumber\\
	&\left({\int_o^R\int_0^{2\pi}}\prod_{k=1}^K \frac{\Gamma\left(\upsilon_k,\frac{\beta_{e,k}}{K\rho\tau_k}\right)}{\Gamma\left(\upsilon_k\right)}{r}drd\theta\right)^J
	\end{align}
	where $\upsilon_k$ and $\tau_k$ are given in \eqref{upsilon_k}.
\end{theorem}
\begin{proof}
	Recalling the PDF of $X_k$ in \eqref{pdf_gamma}, the inner probability in \eqref{ps_equal} can be given by 
	\begin{equation}\label{inner_pro}
	\mathbb{P}\left\{X_k>\frac{\beta_{e,k}}{K\rho}\right\}=\frac{1}{\Gamma\left(\upsilon_k\right)}\Gamma\left(\upsilon_k,\frac{\beta_{e,k}}{K\rho\tau_k}\right).
	\end{equation}
	Conditioned on having $J$ eavesdroppers in the disc $\mathcal{B}_{o,R}$, the distribution of the eavesdroppers' locations follows a binomial point process (BPP). Using the i.i.d. property of a BPP, the joint PDF of the distances ${\bm r_k}=[r_{1,k},\cdots,r_{J,k}]^{\rm T}$ and angles ${\bm \theta_k}=[\theta_{1,k},\cdots,\theta_{J,k}]^{\rm T}$ is given by
	\begin{equation}\label{bpp}
	f_{{\bm r_k, \bm \theta_k}}(r_1,\cdots,r_J,\theta_{1},\cdots,\theta_{J}) = \prod_{j=1}^J\frac{r_{j}}{\pi R^2}.
	\end{equation}
	Substituting \eqref{ppp_n_j}, \eqref{inner_pro}, and \eqref{bpp} into \eqref{ps_equal} completes the proof.
\end{proof}

Although Theorem \ref{theorem_2} does not give a closed-form expression for the secrecy probability, it yields a general and exact result without requiring time-consuming Monte Carlo simulations. Furthermore, it provides a benchmark for comparison with other approximate results.  

If the distance between any two adjacent SBSs is large enough, the correlation of the distances between an eavesdropper and any two SBSs can be ignored due to the random mobility of eavesdroppers. In other words, the positions of eavesdroppers seen from different SBSs can be regarded as independent PPPs $\Phi_{e,k}$ with the same density $\lambda_e$. In this case, the secrecy probability in \eqref{ps_ot_ce} can be recast as 
\begin{equation}\label{ps_ott_ce}
p_{s,ce}^{\rm OT} = 1 - \prod_{k=1}^K
\mathbb{P}\left\{\sum_{e_j\in\Phi_{e,k}}{K\rho}|h_{j,k}|^2\tilde r_{j,k}^{-\alpha}>{\beta_{e,k}}\right\},
\end{equation}
where $\tilde r_{j,k}$ denotes the distance between the $j$-th eavesdropper and the $k$-th SBS after shifting the coordinate system to place the $k$-th SBS at the origin. Note that $\tilde r_{j,k}$ is distinguished from $r_{j,k}$ given in \eqref{ps_ot_ce}.
Since the expectation over $\Phi_{e,k}$ is moved inside the product, the computation can be greatly simplified. Let $I_k ={K\rho} \sum_{e_j\in\Phi_{e,k}}|h_{j,k}|^2\tilde r_{j,k}^{-\alpha}$. We first give the Laplace transform of $I_k$ in the following lemma.
\begin{lemma}[{\cite[Eqn. (8)]{Haenggi2009Stochastic}}]\label{lemma_3}
	Let $\kappa\triangleq \pi\Gamma(1+2/\alpha)\Gamma(1-2/\alpha)$. The Laplace transform of $I_k$ is
	\begin{equation}\label{laplace_lemma_3}
	\mathcal{L}_{I_k}(s)  =\exp\left(-\kappa\lambda_e ({K\rho}s)^{2/\alpha}\right).
	\end{equation}
\end{lemma}
Similar to Theorem \ref{theorem_1}, a closed-form expression for an upper bound of $p_{s,ce}^{\rm OT}$ is given below.
\begin{theorem}\label{theorem_3}
	The secrecy probability in the OT scheme for the CE case satisfies
	\begin{equation}\label{ps_ott}
	p_{s,ce}^{\rm OT} 
	\lessapprox 1 -\prod_{k=1}^K \sum_{m=0}^M {M\choose m}(-1)^{m}\mathcal{L} _{I_k}\left(\frac{m\xi}{\beta_{e,k}}\right),
	\end{equation}
	where $\xi$ and $M$ have been stated in Theorem \ref{theorem_1}.
\end{theorem}

For the special case with $\alpha=4$, we can further derive a closed-form expression for the exact $p_{s,ce}^{\rm OT} $. The PDF of $I_k$ can be obtained from the inverse Laplace transform of $\mathcal{L}_{I_k}(s) $, i.e.,
\begin{equation}\label{pdf_laplace}
f^{\alpha=4}_{I_k}(x) =\mathcal{L}^{-1}_{I_k}(s) = \frac{\kappa\lambda_e\sqrt{K\rho}}{2\sqrt{\pi}x^{3/2}}\exp\left(-\frac{\kappa^2\lambda_e^2K\rho}{4x}\right).
\end{equation}
Plugging \eqref{pdf_laplace} into \eqref{ps_ott_ce} yields 
\begin{align}\label{ps_ce_exact}
p_{s,ce}^{{\rm OT},\alpha=4} &=1 - \prod_{k=1}^K\int_{\beta_{e,k}}^{\infty}f^{\alpha=4}_{I_k}(x)dx \nonumber\\
&\stackrel{\mathrm{(h)}}= 1-\prod_{k=1}^K{\rm erf}\left(\frac{\kappa\lambda_e}{2}\sqrt{\frac{K\rho}{\beta_{e,k}}}\right),
\end{align}
where step $\mathrm{(h)}$ follows from the substitution ${\kappa^2\lambda_e^2K\rho}/({4x})\rightarrow t^2$. Since ${\rm erf}(z)<1$ increases with $z$, it is apparent that $p_{s,ce}^{\rm OT} $ increases with the number of SBSs $K$ and the redundant rate $R_e$, whereas decreases with the density of eavesdroppers $\lambda_e$ and the normalized SNR $\rho$.

\section{Secure Content Delivery Probability Maximization}

In this section, we jointly design the optimal redundant rate $R_e$ and the optimal caching assignment proportion $\phi$ to maximize the overall SCDP $\mathcal{P}_{scd}$. From the definition of $\mathcal{P}_{scd}$ given in \eqref{def_overall_psct}, we
observe that the problem of maximizing $\mathcal{P}_{scd}$ can be decomposed into two steps: 1) designing the optimal $R_e$ to maximize the SCDP $\mathcal{P}^{\rm S}_{scd} = p^{\rm S}_{c}p^{\rm S}_{s}$ for for each scheme $\rm S\in\{ JT,OT,CM\}$; 2) designing the optimal $\phi$ to
maximize the overall SCDP $\mathcal{P}_{scd}$. In what follows, we perform the optimization procedure step by step.
\subsection{Optimization of Redundant Rate $R_e$}
This subsection determines the optimal redundant rate $R_e$ for a target secrecy rate $R_s$ to maximize the SCDP for scheme $\rm S\in\{\rm JT, OT, CM\}$. For tractability, we focus on the NCE case. The optimization in the CE case can be operated similarly, which however would result in a considerable calculation complexity and a much more sophisticated analysis but provide no significant qualitative difference. 

\subsubsection{JT Scheme}
{The SCDP in this case is defined as the product of the connection probability $p_c^{\rm JT}$ in \eqref{pc_jt} and the secrecy probability in \eqref{ps_jt_nce}, i.e., $\mathcal{P}_{scd}^{\rm JT} =p_c^{\rm JT} p_{s}^{\rm JT}$, which can be written as 
	\begin{equation}\label{psct_jt}
	\mathcal{P}_{scd}^{\rm JT} = \exp\left(-A\beta_t-2\lambda_e\int_0^\infty\int_0^{\pi}e^{-B(r,\theta)\beta_e}rdrd\theta\right),
	\end{equation}
	where $A \triangleq {1}/\left({\rho\sum_{k=1}^Kr_{b,k}^{-\alpha} }\right)$ and $B(r,\theta) ={1}/\left({\rho\sum_{k=1}^Kr_{k}^{-\alpha}}\right) $. 
}
Since we have $R_t = R_s + R_e\Rightarrow \beta_t = \beta_s + (1+\beta_s)\beta_e$, substituting $\beta_t$ into \eqref{psct_jt} yields $\mathcal{P}_{scd}^{\rm JT} = e^{-A\beta_s}e^{-Q(\beta_e)}$ such that maximizing $\mathcal{P}_{scd}^{\rm JT} $ is equivalent to minimizing the auxiliary function $Q(\beta_e)$ given below,
\begin{equation}\label{q_jt}
Q(\beta_e) = A(1+\beta_s)\beta_e+2\lambda_e\int_0^\infty\int_0^{\pi}e^{-B(r,\theta)\beta_e}rdrd\theta.
\end{equation}
Hence, we focus on the following problem, and the solution is given in Theorem \ref{theorem_opt_beta_jt},
\begin{equation}\label{problem_beta_jt}
\min_{\beta_e} Q(\beta_e),~~~{\rm s.t.}~~ \beta_e>0.
\end{equation}
\begin{theorem}\label{theorem_opt_beta_jt}
	$Q(\beta_e)$ is convex on $\beta_e$, and the solution $\beta^{\star}_e$ to problem \eqref{problem_beta_jt} is characterized by 
	\begin{equation}\label{opt_beta_jt}
	\frac{dQ(\beta^{\star}_e)}{d\beta^{\star}_e}=0,
	\end{equation}
	i.e., it is the unique zero-crossing of the derivative ${dQ(\beta_e)}/{d\beta_e}$ given below
	\begin{equation}\label{dq1}
	\frac{dQ(\beta_e)}{d\beta_e}=	A(1+\beta_s)-2\lambda_e\int_0^\infty\int_0^{\pi}B(r,\theta)e^{-B(r,\theta)\beta_e}rdrd\theta.
	\end{equation} 
\end{theorem}
\begin{proof}
	Please refer to Appendix \ref{appendix_opt_beta_jt}.
\end{proof} 

Appendix \ref{appendix_opt_beta_jt} shows that ${dQ(\beta_e)}/{d\beta_e}$ increases from negative to positive as $\beta_e$ increases from zero to infinity. 
Then, the value of $\beta_e^{\star}$ can be efficiently obtained via a bisection search with equation \eqref{opt_beta_jt}. The following corollary develops some insights into the behavior of $\beta^{\star}_e$. 
\begin{corollary}\label{corollary_beta_jt} 
	The optimal $\beta^{\star}_e$ that maximizes $\mathcal{P}_{scd}^{\rm JT}$  increases with the eavesdropper density $\lambda_e$, and decreases with the secrecy rate $R_s$ and the SBS-user distance $r_{b,k}$ for $k\in\mathcal{K}$. 
\end{corollary}
 \begin{proof}
 	Let us take $\lambda_e$ as an example. Denote ${dQ(\beta_e)}/{d\beta_e}$ in \eqref{dq1} as $Q_1(\beta_e)$ such that $Q_1(\beta_e^{\star})=0$.
 	Using the derivative rule for implicit functions with $Q_1(\beta_e^{\star})=0$ yields,
 	\begin{align}
 	\frac{dQ_1(\beta_e^{\star})}{d\lambda_e}&=-\frac{\partial Q_1(\beta_e^{\star})/\partial\lambda_e}{\partial Q_1(\beta_e^{\star})/\partial\beta_e} \nonumber\\
 	&= \frac{\int_0^\infty\int_0^{\pi}B(r,\theta)e^{-B(r,\theta)\beta_e}rdrd\theta}{\lambda_e\int_0^\infty\int_0^{\pi}B^2(r,\theta)e^{-B(r,\theta)\beta_e}rdrd\theta}>0.
 	\end{align}
 	Hence, $\beta^{\star}_e$ increases with $\lambda_e$. The other conclusions can be obtained in a similar way.
 \end{proof}
 
 Corollary \ref{corollary_beta_jt} captures an inherent trade-off between the reliability and secrecy. We should choose a large redundant rate for dense eavesdroppers, whereas we should
 keep redundant rate low for a large target secrecy rate or for a remote user
  	  
 \subsubsection{OT Scheme}
 The SCDP in this case is can be obtained by computing the product of the connection probability $ p_c^{\rm OT}$ in \eqref{pc_ot} and the secrecy probability $ p_s^{\rm OT}$ in \eqref{ps_ot_nce}, which is given below,
 \begin{align}\label{psct_ot}
 \mathcal{P}_{scd}^{\rm OT} = \exp\Bigg(&-\frac{1}{K\rho }\sum_{k=1}^Kr_{b,k}^{\alpha}[\beta_s+(1+\beta_s)\beta_{e,k}]-\nonumber\\
 &2\lambda_e\int_0^\infty\int_0^{\pi}e^{-\frac{1}{K\rho }\sum_{k=1}^Kr_{k}^{\alpha}\beta_{e,k}}rdrd\theta\Bigg).
 \end{align}
{Compared with the JT scheme, the SBSs in the OT scheme can use different redundant rates for maximizing the SCDP.
We assume that the knowledge of SBS-user distances is known at the SBSs and can be exploited to determine the optimal redundant rates at different SBSs.} In order to jointly design the redundant rates $R_{e,k}$ for $k\in\mathcal{K}$, we recast \eqref{psct_ot} into a vector form, i.e., $\mathcal{P}_{scd}^{\rm OT} = e^{-\frac{\beta_s}{K\rho}\|\bm r_b\|_1}e^{-\Omega(\bm \beta_e)}$, where the auxiliary function $\Omega(\bm \beta_e)$ is given as below,
\begin{equation}\label{omega_beta}
\Omega(\bm \beta_e) = \frac{1+\beta_s}{K\rho}{\bm r^{\rm T}_b}\bm \beta_e + \lambda_e\int_0^\infty\int_0^{2\pi}e^{-\frac{1}{K\rho}{\bm r^{\rm T}_e}\bm \beta_e}rdrd\theta,
\end{equation}
with  $\bm{r}_b=[r^{\alpha}_{b,1},\cdots,r^{\alpha}_{b,K}]^{\rm T}\geq  0$, $\bm{r}_e=[r^{\alpha}_{1},\cdots,r^{\alpha}_{K}]^{\rm T}\geq  0$, and $\bm{ \beta}_e=[\beta_{e,1},\cdots,\beta_{e,K}]^{\rm T}\geq 0$. Apparently, to maximize $\mathcal{P}_{scd}^{\rm OT}$ we only need to tackle the following minimization problem,
\begin{equation}\label{problem_vector}
\min_{\bm \beta_e} \Omega(\bm \beta_e),~~~{\rm s.t.} ~~\bm \beta_e
\geq  0.
\end{equation}

We point out that the objective function $\Omega(\bm \beta_e)$ in \eqref{omega_beta} is strictly convex on $\bm\beta_e$ due to the summation of an affine function and an integral with
exponential terms. Generally, problem \eqref{problem_vector} can be numerically resolved using some gradient methods, e.g., Newton's method \cite{Boyd2004Convex}. However, Newton's method requires forming and storing the Hessian matrix repeatedly and the computation of the Newton step requires solving a set of linear equations. 
All these operations will bring the system a huge computational burden and thus resulting in a low system efficiency, particularly when the number of SBSs $K$ goes large. To reduce the computational complexity, we propose to process problem \eqref{problem_vector} through an alternating optimization (AO) as described below.
Denote $\left( \beta^{(n)}_{e,1},\cdots,\beta^{(n)}_{e,K}\right)$ as the AO iterate at the $n$-th iteration, and let $\hat{\bm{\beta}}^{(n)}_{e,k}=\bm{\beta}^{(n)}_e\setminus\beta^{(n)}_{e,k}$. We solve the following $K$ subproblems alternatively to obtain $\left( \beta^{(n)}_{e,1},\cdots,\beta^{(n)}_{e,K}\right)$ for $n=1,2,\cdots$
\begin{equation}\label{sub_problem}
\beta^{(n)}_{e,k} =  \arg\min_{\beta_{e,k}\ge 0}\Omega\left(\hat{\bm{\beta}}^{(n)}_{e,k},\beta_{e,k}\right).
\end{equation}
\begin{lemma}\label{opt_beta_ot_lemma}
	When the values of $\beta_{e,j}$ for $j\neq k$ are given, the  solution to problem \eqref{sub_problem} is 
	\begin{align}
	\beta^{\star}_{e,k} &
	=\begin{cases}
	0,& \frac{dQ\left(\hat{\bm{\beta}}^{(n)}_{e,k},\beta_{e,k}\right)}{d\beta_{e,k}}\big|_{\beta_{e,k}=0}\geq 0 ,\\
	\beta^{\circ}_{e,k},&\rm otherwise ,
	\end{cases}
	\end{align}
	where $\beta^{\circ}_{e,k}$	is the unique zero-crossing of the derivative given below,
	\begin{equation}\label{dq_ek}
	\frac{dQ\left(\hat{\bm{\beta}}^{(n)}_{e,k},\beta_{e,k}\right)}{d\beta_{e,k}} = \frac{(1+\beta_s){r_{b,k}}-2\lambda_e\int_0^\infty\int_0^{\pi}r_{e,k}e^{-\frac{{\bm r^{\rm T}_e}\bm \beta_e}{K\rho}}rdrd\theta}{K\rho}.
	\end{equation}
\end{lemma}
\begin{proof}
	Since $Q\left(\hat{\bm{\beta}}^{(n)}_{e,k},\beta_{e,k}\right)$ in \eqref{omega_beta} is a convex function of $\beta_{e,k}$, it arrives at the minimal value at $\beta_{e,k}=0$ if $\left({dQ\left(\hat{\bm{\beta}}^{(n)}_{e,k},\beta_{e,k}\right)}/{d\beta_{e,k}}\right)|_{\beta_{e,k}=0}\ge0$ or at the zero-crossing of ${dQ\left(\hat{\bm{\beta}}^{(n)}_{e,k},\beta_{e,k}\right)}/{d\beta_{e,k}}$ otherwise. 
\end{proof}
The value of $\beta^{\circ}_{e,k}$  can be efficiently calculated through a bisection search with ${dQ\left(\hat{\bm{\beta}}^{(n)}_{e,k},\beta^{\circ}_{e,k}\right)}/{d\beta^{\circ}_{e,k}}=0$.

Lemma \ref{opt_beta_ot_lemma} suggests that the remote SBS, e.g., with a large distance $r_{b,k}$ such that $\left({dQ\left(\hat{\bm{\beta}}^{(n)}_{e,k},\beta_{e,k}\right)}/{d\beta_{e,k}}\right)|_{\beta_{e,k}=0}\ge0$, should set zero redundant rate.
\begin{algorithm}[!t]
	\caption{AO Algorithm for Problem \eqref{problem_vector}}
	\begin{algorithmic}[1]\label{ao_algorithm}
		\STATE Initialize $n=1$, $\bm \beta^{(0)}_e \geq  0$, and assign $\epsilon$ a sufficiently small positive value, e.g., $\epsilon = 10^{-10}$;
		\STATE Update $\bm\beta^{(n)}_e\leftarrow \bm\beta^{(n-1)}_e$;
		\FOR{$k = 1$ to $K$}
		\IF{$\left({dQ\left(\hat{\bm{\beta}}^{(n)}_{e,k},\beta_{e,k}\right)}/{d\beta_{e,k}}\right)|_{\beta_{e,k}=0}\ge0$}
		\STATE $\beta^{(n)}_{e,k} \leftarrow 0$;
		\ELSE 
		\STATE Calculate $\beta^{(n)}_{e,k}$ through a bisection search with the equation  ${dQ\left(\hat{\bm{\beta}}^{(n)}_{e,k},\beta^{(n)}_{e,k}\right)}/{d\beta^{(n)}_{e,k}}=0$;
		\ENDIF
		\STATE  Update $\bm \beta^{(n)}_e$
		\ENDFOR
		\WHILE{$\left|\left[{\Omega\left(\bm\beta_e^{(n)}\right)}-{\Omega\left(\bm\beta_e^{(n-1)}\right)}\right]/{\Omega\left(\bm\beta_e^{(n-1)}\right)}\right|\ge\epsilon$}
		\STATE Update $n \leftarrow n + 1$;
		\STATE Repeat step 2 to step 10;
		\ENDWHILE
		\STATE Output $\bm\beta_e^{(n)}$
	\end{algorithmic}
\end{algorithm}
With Lemma \ref{opt_beta_ot_lemma}, we summarize the whole AO process in Algorithm \ref{ao_algorithm}. Notably, the proposed AO iterative algorithm produces descending objective values, i.e., $\Omega\left(\bm\beta^{(n)}_e\right)<\Omega\left(\bm\beta^{(n-1)}_e\right),\cdots,<\Omega\left(\bm\beta^{(0)}_e\right)$. {Moreover, it has a theoretically provable guarantee on the global optimality of our solution and its convergence.}
\begin{proposition}
	Every limit point $\bm{\beta}^{\star}_e$ of the iterates $\left\{\bm{\beta}^{(n)}_e\right\}$ generated by the AO process in \eqref{sub_problem} is a Karush-Kuhn-Tucker (KKT) point of the primal problem \eqref{problem_vector}.
\end{proposition}
\begin{proof}
Since the objective function $\Omega(\bm \beta_e)$ in \eqref{problem_vector} is a strictly convex function of $\bm\beta_e$, the KKT conditions are necessary and sufficient for the solution to problem \eqref{problem_vector} \cite{Boyd2004Convex}:
\begin{subequations}
	\begin{align}\label{kkt_min}
	&\triangledown\Omega\left(\bm{\beta}^{\star}_e\right) - \bm{\lambda}^{\star} =  0,\\~&\bm{\beta}^{\star}_e\ge 0,~\bm{\lambda}^{\star}\geq  0,~\lambda^{\star}_k\beta^{\star}_{e,k}= 0,~k \in\mathcal{K},
	\end{align}
\end{subequations}
where $\bm{\lambda}^{\star}$ is the Lagrange multiplier introduced for the inequality constraints given in \eqref{problem_vector}. Note that since $\bm{\lambda}^{\star}$ acts as a slack variable in \eqref{kkt_min}, it actually can be eliminated, leaving
\begin{align}\label{kkt}
&	\bm{\beta}^{\star}_e\geq 0,~\triangledown\Omega\left(\bm{\beta}^{\star}_e\right)  \geq  0,~\beta^{\star}_{e,k}\frac{d\Omega(\bm\beta^{\star}_{e})}{d\beta^{\star}_{e,k}}= 0,~k \in\mathcal{K}.
\end{align}
Now, let us recall Lemma \ref{opt_beta_ot_lemma}, from which we have 
\begin{align}\label{kkt_sub}
{\beta}^{\star}_{e,k}\ge  0,~\frac{dQ\left(\hat{\bm{\beta}}^{(n)}_{e,k},\beta^{\star}_{e,k}\right)}{d\beta^{\star}_{e,k}}\ge 0,~\beta^{\star}_{e,k}\frac{dQ\left(\hat{\bm{\beta}}^{(n)}_{e,k},\beta^{\star}_{e,k}\right)}{d\beta^{\star}_{e,k}}=0.
\end{align}
Evidently, \eqref{kkt_sub} gives the KKT conditions for the $k$-th subproblem \eqref{sub_problem}. 
The combination of the KKT conditions for the $K$ subproblems is exactly the KKT conditions for problem \eqref{problem_vector}.
\end{proof}

For a simplified case where all the SBSs use the same redundant rate $R_e$, a more computation-convenient solution to problem \eqref{problem_vector} can be provided by the following theorem.
\begin{theorem}
	If all the SBSs use the same redundant rate $R_e$, the optimal $\beta^{\star}_e$ that minimizes $Q\left(\beta_{e}\right)$ in \eqref{problem_vector} is 
	the unique zero-crossing of the derivative ${dQ\left(\beta_{e}\right)}/{d\beta_{e}}$ given below 
	\begin{equation}
	\frac{dQ\left(\beta_{e}\right)}{d\beta_{e}} = \frac{1+\beta_s}{K\rho}\|\bm r_b\|_1 - \int_0^\infty\int_0^{2\pi}\frac{\lambda_e\|\bm r_e\|_1}{K\rho}e^{-\frac{{\|\bm r_e\|_1} \beta_e}{K\rho}}rdrd\theta.
	\end{equation}
\end{theorem}
The value of $\beta^{\star}_e$ can be efficiently calculated via a bisection search with equation ${dQ\left(\beta_{e}\right)}/{d\beta_{e}}=0$. Some insights into the solution $\beta^{\star}_e$ that are similar to Corollary \ref{corollary_beta_jt} can be developed.

 \subsubsection{CM Scheme}
The SCDP in the CM scheme has the same expression as in the JT scheme, only with $R_s$ increasing to $\delta R_s$. Therefore, the optimal $\beta_e$ that maximizes the SCDP in the CM scheme shares the same form as in Theorem \ref{theorem_opt_beta_jt}, simply by replacing $\beta_s$ with $2^{\delta R_s}-1$.
 
\subsection{Optimization of  Caching Assignment Proportion $\phi$}
This subsection determines the optimal caching assignment proportion $\phi$ that maximizes the overall SCDP $\mathcal{P}_{scd}$ in \eqref{def_overall_psct} with $p_{tr}^{\rm S}$ for $\rm S\in\{JT, OT, CM\}$ given in \eqref{p_jt} and \eqref{p_cm}. 
Note that the summation of discrete sequence aroused by $p_{tr}^{\rm S}$ hampers the optimization of the overall SCDP. Fortunately, the sum of the Zipf probabilities can be approximated as  \cite{Taghizadeh2013Distributed}
\begin{equation}\label{appr_zipf}
\sum_{n=1}^{M} f_n \approx\frac{M^{1-\gamma}-1}{N^{1-\gamma}-1}.
\end{equation}
Invoking \eqref{appr_zipf} with \eqref{p_jt} and \eqref{p_cm} and plugging the obtained results into \eqref{def_overall_psct} with the integer $\lfloor \phi L\rfloor$ replaced with the continuous quantity $ \phi L$, $\mathcal{P}_{scd}$ can be simplified as a continuous function of $\phi$, 
\begin{equation}\label{overall_psct}
\mathcal{P}_{scd}\approx
\frac{\mathcal{\hat P}_{jo}\phi^{1-\gamma}+\mathcal{\hat P}_{oc}[K-K\phi+\phi]^{1-\gamma}-\mathcal{\hat P}_{jc}L^{\gamma-1}}{L^{\gamma-1}(N^{1-\gamma}-1)}+\mathcal{P}^{\rm CM}_{scd},
\end{equation}
where $\mathcal{\hat P}_{jo}=\mathcal{P}^{\rm JT}_{scd}-\mathcal{P}^{\rm OT}_{scd}$, $\mathcal{\hat P}_{oc}=\mathcal{P}^{\rm OT}_{scd}-\mathcal{P}^{\rm CM}_{scd}$, and $\mathcal{\hat P}_{jc}=\mathcal{P}^{\rm JT}_{scd}-\mathcal{P}^{\rm CM}_{scd}$ denote the SCDP differences, with $\mathcal{P}^{\rm S}_{scd}$ being the SCDP for the scheme $\rm S\in\{\rm JT, OT, CM\}$.
Before proceeding to derive the optimal $\phi$ that maximizes $\mathcal{P}_{scd}$, we give the following lemma.
\begin{lemma}\label{jt_cm}
The	JT scheme gives a larger SCDP than does the CM scheme, i.e., $\mathcal{P}^{\rm JT}_{scd}>\mathcal{P}^{\rm CM}_{scd}$.
\end{lemma}
\begin{proof}
	Note that $\mathcal{P}^{\rm CM}_{scd}$ shares the same expression as $\mathcal{P}^{\rm JT}_{scd}$ in \eqref{psct_jt} only with $\beta_s$ increasing from $2^{R_s}-1$ to $2^{\delta R_s}-1$. Then, to complete the proof we only need to prove that the maximal $\mathcal{P}^{\rm JT}_{scd}$ with the optimal $\beta_e^{\star}$ given in Theorem \ref{theorem_opt_beta_jt} decreases with $\beta_s$. By re-expressing $\mathcal{P}^{\rm JT}_{scd}=e^{-W(\beta_s)}$ where  $W(\beta_s)=A\beta_s+Q(\beta^{\star}_e(\beta_s))$ with $A$ and $Q(\beta^{\star}_e(\beta_s))$ given in \eqref{q_jt}, it is also equivalent to proving that $W(\beta_s)$ is an increasing function of $\beta_s$. The derivative ${dW(\beta_s)}/{d\beta_s}$ is given by
	\begin{equation}
	\frac{dW(\beta_s)}{d\beta_s} = A(1+\beta^{\star}_e)+\frac{d\beta^{\star}_e}{d\beta_s}  \frac{dQ(\beta^{\star}_e)}{d\beta^{\star}_e}.
	\end{equation}
From \eqref{opt_beta_jt} we know that ${dQ(\beta^{\star}_e)}/{d\beta^{\star}_e}=0$. Hence, we have ${dW(\beta_s)}/{d\beta_s}>0$.
\end{proof}
\begin{theorem}\label{opt_phi_theorem}
With the proposed hybrid caching placement strategy, the optimal proportion $\phi$  that maximizes $\mathcal{P}_{scd}$ in \eqref{overall_psct} is given by
\begin{align}\label{po}
\phi^{\star}
=\begin{cases}
1,&\mathcal{\hat P}_{jo}>(K-1)\mathcal{\hat P}_{oc},\\
0, &\mathcal{\hat P}_{jo}<0,\\
\frac{1}{1+\frac{1}{K}\left[\left(\frac{(K-1)\mathcal{\hat P}_{oc}}{\mathcal{\hat P}_{jo}}\right)^{\frac{1}{\gamma}}-1\right]},&\rm otherwise.
\end{cases}
\end{align}
\end{theorem}
\begin{proof}
	Please refer to Appendix \ref{appendix_opt_phi}.
\end{proof}

Theorem \ref{opt_phi_theorem} shows that  the SCDP difference between different transmission schemes is critical to the optimal caching assignment. Specifically, when $\mathcal{P}^{\rm CM}_{scd}$ exceeds $\mathcal{P}^{\rm OT}_{scd}$ or when $\mathcal{P}^{\rm JT}_{scd}-\mathcal{P}^{\rm OT}_{scd}$ is $K-1$ times larger than $\mathcal{P}^{\rm OT}_{scd}-\mathcal{P}^{\rm CM}_{scd}$, we have $\phi^{\star} = 1$, meaning that caching the MPFs is more conductive. As $\mathcal{P}^{\rm OT}_{scd}$ increases, the optimal $\phi^{\star}$ becomes smaller, i.e., a larger proportion of the cache unit should be assigned for the DSFs.
We also can prove that the optimal $\phi^{\star}$ increases with the content popularity skewness $\gamma$. The reason behind is that as the content popularity becomes more concentrated (i.e., a larger $\gamma$), the benefit of caching different files becomes limited.
 
\section{Simulation Results}
\begin{figure}[!t]
	\centering
	\includegraphics[width = 3.0in]{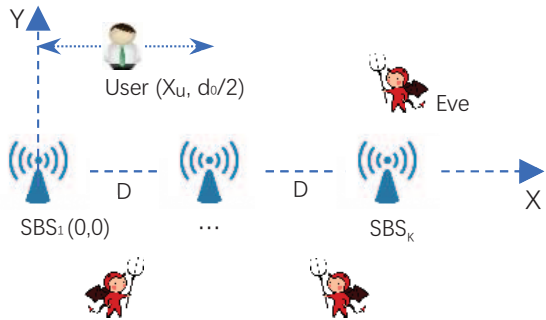}
	\caption{Model used for experiments. $K$ SBSs are placed along the horizontal axis with an identical distance $D$, and the $k$-th nearest SBS to the origin is located at $(0,(k-1)D)$. The user moves along the horizontal direction with a location $(X_u,d_0/2)$. Eavesdroppers are randomly distributed according to a PPP. }
	\label{SIM_SET}
\end{figure}
In this section, we present simulation results to validate our
theoretical analysis. For simplicity, we consider a two-dimensional system model as illustrated in Fig. \ref{SIM_SET}. {Without loss of generality, we set a reference distance $d_0=100$ m and a reference density $\lambda_0=10^{-6}$ ${\rm {nodes/  m}^2}$. Unless specified otherwise, we fix the normalized SNR $\rho = {P}/{(WN_0)}=10$ dB, the path-loss exponent $\alpha = 4$, the distance between two adjacent SBSs $D=d_0$, the content library size $N=100$, and the SBS caching capacity $L=20$. }

\begin{figure}[!t]
	\centering
	\includegraphics[width = 3.0in]{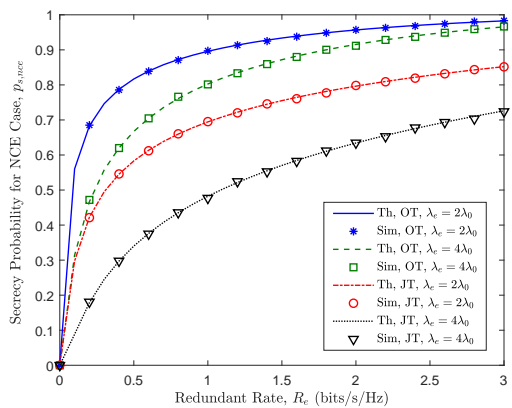}
	\caption{Secrecy probability $p_{s,nce}$ vs. $R_{e}$ for different values of $\lambda_e$, with $K = 3$.}
	\label{PS_NCE}
\end{figure}
\begin{figure}[!t]
	\centering
	\includegraphics[width = 3.0in]{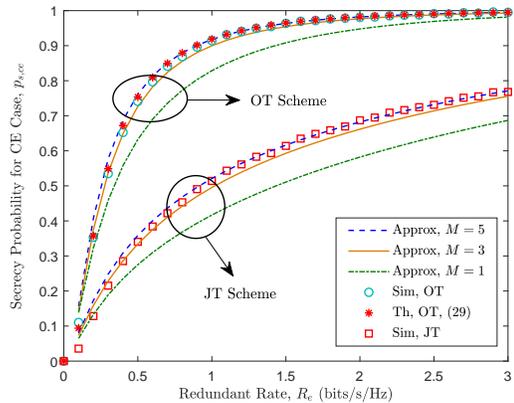}
	\caption{Secrecy probability $p_{s,ce}$ vs. $R_{e}$ for different values of $M$, with $K = 3$, $D = 6d_0$, and $\lambda_e = 3\lambda_0$.}
	\label{PS_CE}
\end{figure}
Fig. \ref{PS_NCE} and Fig. \ref{PS_CE} depict the secrecy probabilities versus the redundant rates $R_e$ in the NCE and CE cases, respectively. The Monte-Carlo simulation results match well with the theoretical values. Both figures verify the superiority of the OT scheme over the JT scheme in terms of transmission secrecy. As expected, the secrecy probability increases with $R_e$ and decreases with $\lambda_e$. Fig. \ref{PS_CE} shows that the approximate results given in \eqref{ps_jt_ce} and \eqref{ps_ott} coincide well with the real ones when $M = 5$. We also find that the results in \eqref{ps_ce_exact} approach the simulated ones in \eqref{ps_oft}, particularly for a large distance $D$, e.g., $D=6d_0$. This is because, for a sufficiently large $D$, the correlation in the distances between an eavesdropper and any two SBSs can be ignored.

\begin{figure}[!t]
	\centering
	\includegraphics[width = 3.0in]{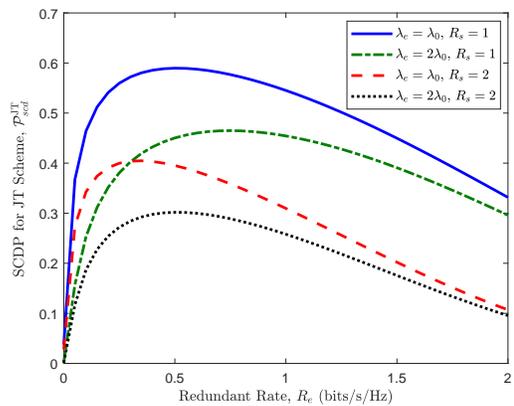}
	\caption{SCDP $\mathcal{P}_{scd}^{\rm JT}$ vs. $R_{e}$ for different values of $\lambda_e$ and $R_s$, with $K = 3$, and $X_u = 3d_0$.}
	\label{SCDP_JT}
\end{figure}
\begin{figure}[!t]
	\centering
	\includegraphics[width = 3.0in]{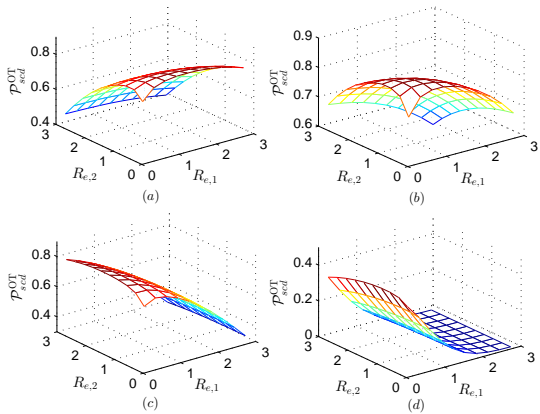}
	\caption{SCDP $\mathcal{P}_{scd}^{\rm OT}$ vs. $R_{e,1}$ and $R_{e,2}$ for  $X_u=\{0.3,0.5,0.8,1.5\}d_0$ in $(a)$-$(d)$, with $K = 2$, $\lambda_e = \lambda_0$, and $R_s = 1$ bit/s/Hz.}
	\label{SCDP_OT}
\end{figure}
Fig. \ref{SCDP_JT} and Fig. \ref{SCDP_OT} plot the SCDP $\mathcal{P}_{scd}^{\rm S}$ as a function of the redundant rate $R_e$ in the JT and OT schemes, respectively. Fig. \ref{SCDP_JT} shows that the SCDP $\mathcal{P}_{scd}^{\rm JT}$ in the JT scheme first increases and then decreases with $R_e$, just as proved in Theorem \ref{theorem_opt_beta_jt}. We also find that the optimal $R_e$ that maximizes $\mathcal{P}_{scd}^{\rm JT}$ increases with a larger $\lambda_e$ or a smaller $R_s$, verifying the insights obtained in Corollary \ref{corollary_beta_jt}. Fig. \ref{SCDP_OT} shows how the SCDP $\mathcal{P}_{scd}^{\rm OT}$ in the OT scheme varies with the redundant rates $R_{e,k}$ at different SBSs when the user is located differently. Sub-figures $(a)$ to $(d)$ show that $\mathcal{P}_{scd}^{\rm OT}$ is dominated by the redundant rate at the closer SBS to the user. This suggests that in order to improve the SCDP, we should set zero redundant rate at those remote SBSs and should carefully design the redundant rates at those neighboring SBSs, just as indicated in Lemma \ref{opt_beta_ot_lemma}.

\begin{figure}[!t]
	\centering
	\includegraphics[width = 3.0in]{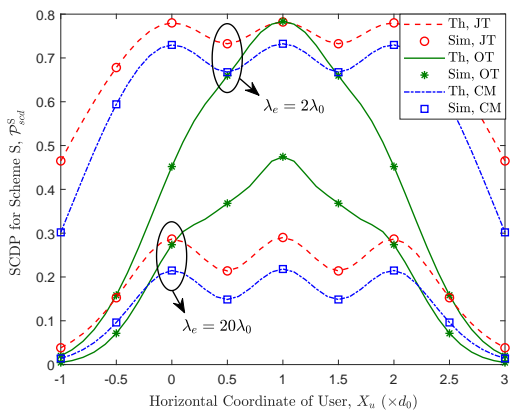}
	\caption{SCDP $\mathcal{P}_{scd}^{\rm S}$ for $\rm S\in\{\rm JT, OT, CM\}$ vs. $X_u$ for different values of $\lambda_e$, with $K = 3$, $\delta = 2$, and $R_s = 1$ bit/s/Hz.}
	\label{SCDP_REAL_B_LE}
\end{figure}
Fig. \ref{SCDP_REAL_B_LE} illustrates how the SCDP is affected by the geographical relationship between the user and the SBSs in various transmission schemes. With the experimental model described in Fig. \ref{SIM_SET} in mind, we find that for the JT scheme, if only the user moves close to one of the SBSs, the SCDP can be remarkably improved. Whereas for the OT scheme, only if the user is located approximately at the center of the SBSs, or there is no significant difference among the distances between the user and different SBSs, a high SCDP can be achieved; otherwise the SCDP performance is severely degraded. In addition, Fig. \ref{SCDP_REAL_B_LE} shows that the JT scheme outperforms the CM scheme, just as proved in Lemma \ref{jt_cm}. However, whether the JT or the OT scheme is superior depends on the specific transmission environment. For example, as can be seen from Fig. \ref{SCDP_REAL_B_LE}, in a sparse eavesdropper scenario, the JT scheme outperforms the OT scheme; whereas in a dense eavesdropper case and when the user is located at the center of the SBSs, the OT scheme provides a higher benefit for the secure content delivery. 

\begin{figure}[!t]
	\centering
	\includegraphics[width = 3.0in]{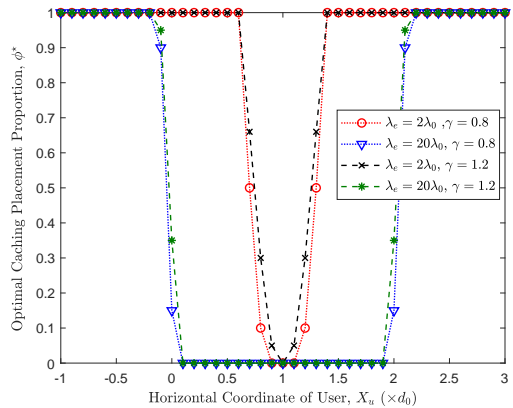}
	\caption{Optimal caching proportion $\phi^{\star}$ vs. $X_u$ for different values of $\lambda_e$ and $\gamma$, with $K = 3$, $\delta = 2$, and $R_s = 1$ bit/s/Hz.}
	\label{OPT_PHI_REAL_B_LE_GAMMA}
\end{figure}
Fig. \ref{OPT_PHI_REAL_B_LE_GAMMA} shows how the optimal caching placement is influenced by the user's location and the eavesdropper density. {When eavesdroppers are distributed sparsely (a small $\lambda_e$) or when the user is located far away from the SBSs, the SCDP is bottlenecked by the connection probability. This suggests that the SBSs should adopt the JT scheme to improve transmission reliability. Hence, the optimal $\phi^{\star}$ goes to one, meaning that caching the MPFs is more beneficial. When the eavesdropper density $\lambda_e$ increases or the user moves close to the SBSs, the secrecy probability becomes the major bottleneck for improving the SCDP. In order to guarantee a certain level of secrecy, the OT scheme is preferred. Hence, the optimal $\phi^{\star}$ approaches zero, suggesting that the SBSs should store the DSFs geographically. For a moderate density of eavesdroppers or SBS-user distance, the optimal $\phi^{\star}$ lies between zero and one, showing that there exists a trade-off between caching the MPFs and the DSFs.} We also find that the optimal $\phi^{\star}$ increases with the content popularity skewness $\gamma$, which coincides with the finding given in Theorem \ref{opt_phi_theorem}.

\begin{figure}[!t]
	\centering
	\includegraphics[width = 3.0in]{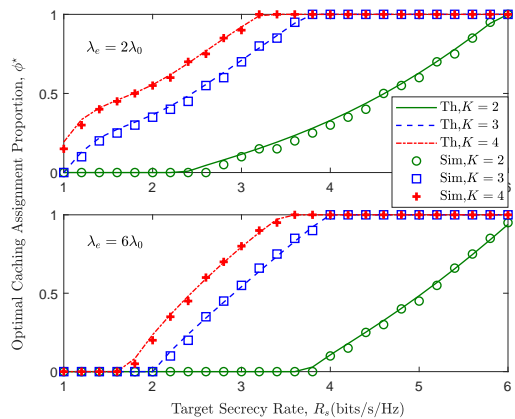}
	\caption{Optimal caching assignment proportion $\phi^{\star}$ vs. $R_s$ for different $K$ and $\lambda_e$, with $\gamma=1.2$, and $\delta = 3$.}
	\label{OPT_PHI_RS_LE_K}
\end{figure}
Fig. \ref{OPT_PHI_RS_LE_K} depicts the optimal caching proportion $\phi^{\star}$ versus the target secrecy rate $R_s$. Monte-Carlo simulated results match well with the theoretical values, verifying the accuracy of the approximation in \eqref{overall_psct}. We show that $\phi^{\star}$ initially is small at the small $R_s$ region and then increases with $R_s$. This is because, to support a large $R_s$, transmission reliability should be adequately ensured and thus the JT scheme becomes favorable. We also observe that $\phi^{\star}$ increases with the number of SBSs $K$. This means, with more cooperative SBSs, the benefit from caching the MPFs along with the JT scheme would be more pronounced.

\begin{figure}[!t]
	\centering
	\includegraphics[width = 3.0in]{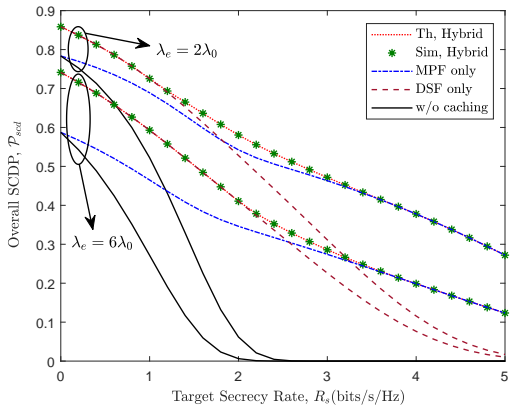}
	\caption{Overall SCDP $\mathcal{P}_{scd}$ vs. $R_s$ for different $\lambda_e$, with $K=3$, $\gamma=1.2$, and $\delta = 3$.}
	\label{SCDP_RS_LE}
\end{figure}
Fig. \ref{SCDP_RS_LE} compares the overall SCDPs for different caching strategies. We show that the proposed hybrid caching strategy outperforms either the MPF-only or the DSF-only scheme. In particular, hybrid caching can provide a remarkable performance gain over MPF-only caching for a small secrecy rate $R_s$ and a large eavesdropper density $\lambda_e$, or over DSF-only caching for a large $R_s$ and a small $\lambda_e$.
In addition, the SCDP performance without caching severely deteriorates.

\section{Conclusions and Future Work}
This paper studies the security issue in a cache-enabled cooperative SCN against randomly located eavesdroppers. We propose a hybrid MPF and DSF caching strategy along with the JT and OT schemes. We derive analytical expressions for the SCDP in each transmission scheme for both NCE and CE scenarios, based on which the optimal transmission rates and the optimal caching assignment are jointly designed for maximizing the overall SCDP. We also develop various insights into the optimal designs. Numerical results demonstrate the superiority of the proposed hybrid caching strategy over the MPF- and DSF-only ones in terms of the SCDP.

This paper opens up several interesting research directions.
{For example, the proposed analysis and design framework can be extended to investigate the cooperative multi-antenna SBSs in cache-enabled SCNs, where artificial jamming signals can be exploited to confound eavesdroppers.} The potential of PLS can be further tapped by performing adaptive designs leveraging the instantaneous CSI of the main channels. Studying the secure content delivery from a network perspective, e.g., considering a multi-cell cellular network, is also an interesting issue. Nevertheless, this might be much more sophisticated, since we should analyze the influence of both random interferers and eavesdroppers. Another possible direction is to consider diverse secrecy attributes for different files, and to explore more intelligent caching strategies.

\appendix
\subsection{Proof of Lemma \ref{lemma_2}}
\label{appendix_lemma_2}
The $i$-th cumulant of the RV $X_k$ is defined as 
\begin{equation}\label{cumulant}
Q_{X_k}^{(i)} = \frac{d^i\mathbb{E}_{X_k}  \left[e^{\omega X_k}\right]}{d\omega^i}\big|_{\omega = 0}.
\end{equation} 
The mean and variance of $X_k$ are $\mu_{X_k} = Q_{X_k}^{(1)}$ and $\sigma_{X_k}^2 = Q_{X_k}^{(2)} - \left(Q_{X_k}^{(1)}\right)^2$. Hence, the parameters $\upsilon_k$ and $\tau_k$ in \eqref{pdf_gamma} can be calculated as $
\upsilon_k ={\mu_{X_k}^2}/{\sigma_{X_k}^2}$ and $\tau_k={\sigma_{X_k}^2}/{\mu_{X_k}}$.
{Recall that $X_k = \sum_{e_j\in\Phi_{e}^R}|h_{j,k}|^2r_{j,k}^{-\alpha}$. Due to the mutual independence among $\{h_{j,k}\}$ for $e_j\in\Phi_{e}^R$, we can express $\mathbb{E}_{X_k}\left[e^{\omega X_k}\right]$ for a fixed $\Phi_{e}^R$ as follows,
\begin{equation}
\mathbb{E}_{X_k}\left[e^{\omega X_k}\right]=\prod_{e_j\in\Phi_{e}^R}\mathbb{E}_{h_{j,k}}\left[e^{-\omega |h_{j,k}|^2r_{j,k}^{-\alpha}}\right].
\end{equation}
Denote $g_{j,k}(\omega)=\mathbb{E}_{h_{j,k}}\left[e^{-\omega |h_{j,k}|^2r_{j,k}^{-\alpha}}\right]$ such that $\mathbb{E}_{X_k}\left[e^{\omega X_k}\right]=\prod_{e_j\in\Phi_{e}^R}g_{j,k}(\omega)$ and $Q_{X_k}^{(i)} = \left({d^ig_{j,k}(\omega)}/{d\omega^i}\right)|_{\omega = 0}$.
Then, we have  $g_{j,k}(\omega)|_{\omega=0}=1$, $\left({dg_{j,k}(\omega)}/{d\omega}\right)|_{\omega=0}=r_{j,k}^{-\alpha}$, and $\left({d^2g_{j,k}(\omega)}/{d\omega^2}\right)|_{\omega=0}=2r_{j,k}^{-2\alpha}$. Based on these results, the first and second cumulants of $X_k$ respectively can be calculated as
\begin{equation}\label{cumulant_1}
Q^{(1)}_{X_k}= \sum_{e_j\in\Phi_{e}^R}\prod_{e_l\in\Phi_{e}^R\setminus{e_j}}\left(g_{l,k}(\omega)\frac{dg_{j,k}(\omega)}{d\omega}\right)\big|_{\omega=0} = \sum_{e_j\in\Phi_{e}^R}r_{j,k}^{-\alpha},
\end{equation}
\begin{align}\label{cumulant_2}
Q^{(2)}_{X_k}&= \sum_{e_j\in\Phi_{e}^R}\left[\prod_{e_l\in\Phi_{e}^R\setminus{e_j}}\left(g_{l,k}(\omega)\frac{d^2g_{j,k}(\omega)}{d\omega^2}\right)+\frac{dg_{j,k}(\omega)}{d\omega}\times\right.\nonumber\\
&  \quad\quad\quad\left.\sum_{e_l\in\Phi_{e}^R\setminus{e_j}}\prod_{e_q\in\Phi_{e}^R\setminus{\{e_j,e_l\}}}\left(g_{q,k}(\omega)\frac{dg_{l,k}(\omega)}{d\omega}\right)\right]\bigg|_{\omega=0}\nonumber\\
& = \sum_{e_j\in\Phi_{e}^R}\left[2r_{j,k}^{-2\alpha} + r_{j,k}^{-\alpha}\sum_{e_l\in\Phi_{e}^R\setminus{e_j}} r_{l,k}^{-\alpha}  \right].
\end{align}
}After some algebraic manipulations, we can obtain $\mu_{X_k} = \sum_{e_j\in\Phi_{e}^R}r_{j,k}^{-\alpha}$ and $\sigma_{X_k}^2 = \sum_{e_j\in\Phi_{e}^R}r_{j,k}^{-2\alpha}$, and substituting them into the expressions of $\upsilon_k$ and $\tau_k$ completes the proof.


\subsection{Proof of Theorem \ref{theorem_opt_beta_jt}}
\label{appendix_opt_beta_jt}
Recall \eqref{dq1}, and the  second derivative ${d^2Q(\beta_e)}/{d\beta_e^2}$ can be given by
\begin{equation}\label{dq2}
\frac{d^2Q(\beta_e)}{d\beta_e^2} = \lambda_e\int_0^\infty\int_0^{2\pi}B^2(r,\theta)e^{-B(r,\theta)\beta_e}rdrd\theta>0.
\end{equation}
This means that $Q(\beta_e)$ is strictly convex on $\beta_e$. Next, we determine the signs of ${dQ(\beta_e)}/{d\beta_e}$ at $\beta_e\rightarrow\infty$ and $\beta_e = 0$. We can prove that $\left({dQ(\beta_e)}/{d\beta_e}\right)|_{\beta_e\rightarrow \infty}=A(1+\beta_s)>0$ from \eqref{dq1}, and $\left({dQ(\beta_e)}/{d\beta_e}\right)|_{\beta_e=0}<0$ by realizing that $\mathcal{P}_{scd}|_{\beta_e=0}=0$ and $\mathcal{P}_{scd}|_{\beta_e>0}>0$ from \eqref{psct_jt}, respectively. Since ${dQ(\beta_e)}/{d\beta_e}$ increases with $\beta_e$, there exists a unique $\beta_e$ that satisfies ${dQ(\beta_e)}/{d\beta_e}=0$, which is the solution to  \eqref{problem_beta_jt}.
	
\subsection{Proof of Theorem \ref{opt_phi_theorem}}
\label{appendix_opt_phi}
Recall \eqref{overall_psct}, and the first and second derivatives of $\mathcal{P}_{scd}$ respectively can be given by
\begin{equation}\label{dP1}
\frac{d\mathcal{P}_{scd}}{d\phi} =\frac{\phi^{-\gamma}\mathcal{\hat P}_{jo}-(K-1)[K-\phi(K-1)]^{-\gamma} \mathcal{\hat P}_{oc}}{\Delta}  ,
\end{equation}
\begin{equation}\label{dP2}
\frac{d^2\mathcal{P}_{scd}}{d\phi^2} = -\frac{\phi^{-\gamma-1}\mathcal{\hat P}_{jo}+(K-1)^2[K-\phi(K-1)]^{-\gamma-1}\mathcal{\hat P}_{oc}}{\Delta\gamma^{-1}},
\end{equation}
where $\Delta \triangleq {L^{\gamma-1}\left(N^{1-\gamma}-1\right)}/{(1-\gamma)}>0$. Since $\mathcal{P}^{\rm JT}_{scd}>\mathcal{P}^{\rm CM}_{scd}$ always holds, we can determine the optimal $\phi$ that
maximizes $\mathcal{P}_{scd}$ by distinguishing three cases:
1) If $\mathcal{P}^{\rm OT}_{scd}<\mathcal{P}^{\rm CM}_{scd}$, we have ${d\mathcal{P}_{scd}}/{d\phi}>0$, i.e., $\mathcal{P}_{scd}$ monotonically increases with $\phi$. Hence, $\mathcal{P}_{scd}$ reaches the maximal value at $\phi^{\star}=1$;
2) If $\mathcal{P}^{\rm OT}_{scd}>\mathcal{P}^{\rm JT}_{scd}$, we have ${d\mathcal{P}_{scd}}/{d\phi}<0$, i.e., $\mathcal{P}_{scd}$ monotonically decreases with $\phi$. Then, the maximal $\mathcal{P}_{scd}$ is achieved at  $\phi^{\star}=0$;
3) If $\mathcal{P}^{\rm JT}_{scd}\ge\mathcal{P}^{\rm OT}_{scd}\ge\mathcal{P}^{\rm CM}_{scd}$, we have ${d^2\mathcal{P}_{scd}}/{d\phi^2}<0$, i.e., $\mathcal{P}_{scd}$ is concave on $\phi$. Since $\left({d\mathcal{P}_{scd}}/{d\phi}\right)|_{\phi=0}>0$, $\mathcal{P}_{scd}$ arrives at the maximal value at $\phi=1$ if $\left({d\mathcal{P}_{scd}}/{d\phi}\right)|_{\phi=1}>0$ or otherwise at the zero-crossing of the derivative ${d\mathcal{P}_{scd}}/{d\phi}$. 
By now, we have completed the proof.


\begin{thebibliography}{99}
	
	\bibitem{Wang2014Cache}
	X. Wang, M. Chen, T. Taleb, A. Ksentini, and V. Leung, ``Cache in the air: Exploiting content caching and delivery techniques for 5G systems,'' \emph{IEEE Commun. Mag.}, vol. 52, no. 2, pp. 131--139, Feb. 2014.
	
		\bibitem{Liu2014Cache}
		A. Liu and V. Lau, ``Cache-enabled opportunistic cooperative MIMO
		for video streaming in wireless systems,'' \emph{IEEE Trans. Signal Process.}, vol. 62, no. 2, pp. 390--402, Jan. 2014.
		
		\bibitem{Xiang2018Cache}
		L. Xiang, D. W. K. Ng, R. Schober, and V. W. S. Wong, ``Cache-enabled physical layer security for video streaming in backhaul-limited cellular networks,'' \emph{IEEE Trans. Wireless Commun.}, vol. 17, no. 2, pp. 736--751, Feb. 2018.
		
	\bibitem{HTTP}
	E. Rescorla, ``HTTP over TLS,'' \emph{IETF RFC 2818}, May 2000.
	
	
	\bibitem{Paschos2016Wireless}
	G. Paschos, E. Ba\c{s}tu\u{g}, I. Land, G. Caire, and M. Debbah, ``Wireless
	caching: Technical misconceptions and business barriers,'' \emph{IEEE Commun. Mag.}, vol. 54, no. 8, pp. 16--22, Aug. 2016.
	
	
		\bibitem{Wyner1975Wire-tap}
		A. D. Wyner, ``The wire-tap channel,'' \emph{Bell Syst. Tech. J.}, vol. 54, no. 8, pp. 1355--1387, Oct. 1975.
	
		\bibitem{Wang2016Physical_book}
		H.-M. Wang and T.-X. Zheng, \emph{Physical Layer Security in Random Cellular Networks}. Singapore: Springer Press, 2016.
		
	\bibitem{Yang2015Safeguading}
	N. Yang, L. Wang, G. Geraci, M. Elkashlan, J. Yuan, and M. D. Renzo, ``Safeguarding 5G wireless communication networks using physical tier security,'' \emph{IEEE Commun. Mag.}, vol. 53, no. 4, pp. 20--27, Apr. 2015.
	
	\bibitem{Wang2015Enhancing}
	H.-M. Wang and X.-G. Xia, ``Enhancing wireless secrecy via cooperation: Signal design and optimization,'' \emph{IEEE Commun. Mag.}, vol. 53, no. 12,
	pp. 47--53, Dec. 2015.
	
	
	\bibitem{Zou2016Survey}
	Y. Zou, X. Wang, and L. Hanzo, ``A survey on wireless security: Technical challenges, recent advances and future trends,'' \emph{Proc. IEEE}, vol. 104, no. 9, pp. 1727--1765, Sep. 2016.
	
	
	
	%
	\bibitem{Zhou2011Throughput}
	X. Zhou, R. Ganti, J. Andrews, and A. Hj{\o}rungnes, ``On the throughput cost of physical tier security in decentralized wireless networks,'' \emph{IEEE Trans. Wireless Commun.}, vol. 10, no. 8, pp. 2764--2775, Aug. 2011.
	%
	\bibitem{Zhang2013Enhancing}
	X. Zhang, X. Zhou, and M. R. McKay, ``Enhancing secrecy with multi-antenna transmission in wireless ad hoc networks,'' \emph{IEEE Trans. Inf. Forensics and Security}, vol. 8, no. 11, pp. 1802--1814, Nov. 2013.
	
	
	\bibitem{Zheng2015Multi}
	T.-X. Zheng, H.-M. Wang, J. Yuan, D. Towsley, and M. H. Lee, ``Multi-antenna transmission with artificial noise against randomly distributed eavesdroppers,'' \emph{IEEE Trans. Commun.}, vol. 63, no. 11, pp. 4347--4362, Nov. 2015.
	
	\bibitem{Zheng2017Safeguarding}
	T.-X. Zheng, H.-M. Wang, Q. Yang, and M. H. Lee, ``Safeguarding decentralized wireless networks using full-duplex jamming receivers,'' \emph{IEEE Trans. Wireless Commun.}, vol. 16, no. 1, pp. 278--292, Jan. 2017.
	
	\bibitem{Zheng2017Physical}
	T.-X. Zheng, H.-M. Wang, J. Yuan, Z. Han, and M. H. Lee, ``Physical layer security in wireless ad hoc networks under a hybrid full-/half-duplex receiver deployment strategy,'' \emph{IEEE Trans. Wireless Commun.}, vol. 16, no. 6, pp. 3827--3839, Jun. 2017. 
	
	\bibitem{Wang2016Physical}
	H.-M. Wang, T.-X. Zheng, J. Yuan, D. Towsley, and M. H. Lee,
	``Physical layer security in heterogeneous cellular networks,'' \emph{IEEE Trans. Commun.}, vol. 64, no. 3, pp. 1204--1219, Mar. 2016.
%
	
	\bibitem{Deng2016Artificial}
	Y. Deng, L. Wang, S. A. R. Zaidi, J. Yuan, and M. Elkashlan, ``Artificial-noise aided secure transmission in large scale spectrum sharing networks,'' \emph{IEEE Trans. Commun.}, vol. 64, no. 5, pp. 2116--2129, May 2016.
	
	\bibitem{Liu2017Enhancing}
	Y. Liu, Z. Qin, M. Elkashlan, Y. Gao, and L. Hanzo, ``Enhancing the
	physical layer security of non-orthogonal multiple access in large-scale networks,'' \emph{IEEE Trans. Wireless Commun.}, vol. 16, no. 3, pp. 1656--1672, Mar. 2017.

	\bibitem{Shanmugam2013FemtoCaching}
	K. Shanmugam, N. Golrezaei, A. Dimakis, A. Molisch, and G. Caire,
	``FemtoCaching: Wireless content delivery through distributed caching
	helpers,'' \emph{IEEE Trans. Inf. Theory}, vol. 59, no. 12, pp. 8402--8413, Dec. 2013.
	
	
	\bibitem{Liu2014Will}
	D. Liu and C. Yang, ``Will caching at base station improve energy
	efficiency of downlink transmission?'' in \emph{Proc. IEEE GlobalSIP}, Atlanta,
	GA, Dec. 2014.
	
	
	\bibitem{Xiang2017Cross}
	L. Xiang, D. W. K. Ng, T. Islam, R. Schober, V. W. S. Wong, and J. Wang, ``Cross-layer optimization of fast video delivery in cache- and buffer-enabled relaying networks,'' \emph{IEEE Trans. Veh. Tech.}, vol. 66, no. 12, pp. 11366--11382, Jun. 2017.
	
	\bibitem{Chae2015Cooperative}
	S. H. Chae, J. Y. Ryu, T. Q. S. Quek, and W. Choi, ``Cooperative transmission via caching helpers,'' in \emph{Proc. IEEE GLOBECOM}, San Diego, CA, USA, Dec. 2015, pp. 1–-6
	
	
	
	\bibitem{Ao2015Distributed}
	W. C. Ao and K. Psounis, ``Distributed caching and small cell cooperation for fast content delivery,'' in \emph{Proc. ACM MobiHoc}, Hangzhou, China, Jun. 2015, pp. 127--136.
	
	\bibitem{Peng2015Backhaul}
	X. Peng, J.-C. Shen, J. Zhang, and K. B. Letaief, ``Backhaul-aware
	caching placement for wireless networks,'' in \emph{Proc. IEEE GLOBECOM}, San Diego, CA, USA, Dec. 2015, pp. 1--6.
	
	\bibitem{Chen2017Cooperative}
	Z. Chen, J. Lee, and M. Kountouris, ``Cooperative caching and transmission design in cluster-centric small cell networks,'' \emph{IEEE Trans. Wireless Commun.}, vol. 16, no. 5, pp. 3401--3415, Mar. 2017.
	
		
		\bibitem{Maddah-Ali2014Fundamental}
		M. A. Maddah-Ali and U. Niesen, ``Fundamental limits of caching,'' \emph{IEEE Trans. Inf. Theory},
		vol. 60, no. 5, pp. 2856--2867, May 2014.
		
	
	\bibitem{Sengupta2015Fundamental}
	A. Sengupta, R. Tandon, and T. C. Clancy, ``Fundamental limits of caching with secure delivery,'' \emph{IEEE Trans. Inf. Forensics and Security}, vol. 10, no. 2,
	pp. 355--370, Feb. 2015.
	
	\bibitem{Awan2015Fundamental}
	Z. H. Awan and A. Sezgin, ``Fundamental limits of caching in D2D networks with secure delivery,'' in \emph{Proc. IEEE ICC Workshops}, London, UK, Jun. 2015.
	
	\bibitem{Gerami2015Secure}
	M. Gerami, M. Xiao, S. Salimi, and M. Skoglund, ``Secure partial repair in wireless caching networks with broadcast channels,'' in \emph{Proc. IEEE Conf. CNS}, Sep. 2015, pp. 353--360.
	
	\bibitem{Gabry2016On}
	F. Gabry, V. Bioglio, and I. Land, ``On edge caching with secrecy constraints,'' in \emph{Proc. IEEE	ICC}, Kuala Lumpur, Malaysia, May 2016.
	

	
{	\bibitem{Xiang2018Secure}
	L. Xiang, D. W. K. Ng, R. Schober, and V. W. S. Wong, ``Secure video streaming in heterogeneous small	cell networks with untrusted cache helpers,'' \emph{IEEE Trans. Wireless Commun.}, vol. 17, no. 4, pp. 2645--2661, Apr. 2018.}

\bibitem{Haenggi2009Stochastic}
M. Haenggi, J. Andrews, F. Baccelli, O. Dousse, and M. Franceschetti, ``Stochastic geometry and random graphs for the analysis and design of wireless networks,'' \emph{IEEE J. Sel. Areas Commun.}, vol. 27, no. 7, pp. 1029--1046, Sep. 2009.
		
	
	\bibitem{Ohm2005Advances}
	J.-R. Ohm, ``Advances in scalable video coding,'' \emph{Proc. IEEE}, vol. 93, no. 1, pp. 42--56, Jan. 2005.

	
	\bibitem{Singh2015Tractable}
	S. Singh, M. N. Kulkarni, A. Ghosh, and J. G. Andrews, ``Tractable model for rate in self-backhauled millimeter wave cellular networks,'' \emph{IEEE J. Sel. Areas Commun.}, vol. 33, no. 10, pp. 2196--2211, Oct. 2015.
	
	
	\bibitem{Alzer1997On}
	H. Alzer, ``On some inequalities for the incomplete gamma function,''
	\emph{Math. Comput.}, vol. 66, no. 218, pp. 771--778, Apr. 1997.
	
{	\bibitem{Zhang2017Energy}
	J. Zhang, L. Xiang, D. W. K. Ng, M. Jo, and M. Chen, ``Energy efficiency evaluation of multi-tier
	cellular uplink transmission under
	maximum power constraint,'' \emph{IEEE Trans. Wireless Commun.}, vol. 16, no, 11, pp. 7092--7107, Nov. 2017.}

	
	\bibitem{Heath2013Modeling}
	R. W. Heath, M. Kountouris, and T. Bai, ``Modeling heterogeneous
	network interference using Poisson point processes,'' \emph{IEEE Trans. Signal Process.}, vol. 61, no. 16, pp. 4114–-4126, Aug. 2013.
	

	
	
	\bibitem{Taghizadeh2013Distributed}
	M. Taghizadeh, K. Micinski, S. Biswas, C. Ofria, and E. Torng,
	``Distributed cooperative caching in social wireless networks,'' \emph{IEEE
		Trans. Mobile Comput.}, vol. 12, no. 6, pp. 1037--1053, Jun. 2013.
	
	
	\bibitem{Chiu2013Stochastic}
	S. N. Chiu, D. Stoyan, W. Kendall, and J. Mecke, \emph{Stochastic Geometry and its Applications, 3rd ed}. John Wiley and Sons, 2013.
	%
	\bibitem{Gradshteyn2007Table}
	I. S. Gradshteyn, I. M. Ryzhik, A. Jeffrey, D. Zwillinger, and S. Technica, \emph{Table of Integrals, Series, and Products, 7th ed}. New York:
	Academic Press, 2007.
	%
	\bibitem{Boyd2004Convex}
	S. Boyd and L. Vandenberghe, \emph{Convex Optimization}. Cambridge, UK: Cambridge University Press, 2004.
	
\end{thebibliography}
\end{document}